\documentclass[11pt]{article} 
\usepackage{fullpage}        
\usepackage{setspace}        
\usepackage{float}           
\usepackage{placeins}         

\usepackage{kotex}            
\usepackage[utf8]{inputenc}  

\usepackage{amsmath}          
\usepackage{amsfonts}         
\usepackage{amssymb}          
\usepackage{amsthm}           
\usepackage{mathtools}        

\usepackage{booktabs}         
\usepackage{caption}
\usepackage[round]{natbib}    

\usepackage{algorithm}        
\usepackage{algpseudocode}    

\usepackage[dvipsnames]{xcolor}
\usepackage[colorlinks=true,  
            linkcolor=blue,
            urlcolor=blue,
            citecolor=blue]{hyperref}

\usepackage{authblk}
 
\newtheorem{theorem}{Theorem} 

\newtheorem{lemma}{Lemma}

\newtheorem{remark}{Remark}

\newcommand{\Expect}{\mathbb{E}}

\newcommand{\bI}{\boldsymbol{I}}
\newcommand{\bS}{\boldsymbol{S}}
\newcommand{\bX}{\boldsymbol{X}}
\newcommand{\bY}{\boldsymbol{Y}}
\newcommand{\bzero}{\boldsymbol{0}}

\newcommand{\bmu}{\boldsymbol{\mu}}

\newcommand{\bSigma}{\boldsymbol{\Sigma}}
\newcommand{\bdelta}{\boldsymbol{\delta}}
\newcommand{\bDelta}{\boldsymbol{\Delta}}

\title{Block-Independent Likelihood Ratio Testing for High-Dimensional Mean Vectors with Applications to Matrix-Variate Data}
\date{}

\author[1,$*$]{Minsub Shin}
\author[2,$*$]{Kwangok Seo}
\author[3,4]{Sang Han Lee}
\author[1]{Johan Lim}

\affil[1]{Department of Statistics, Seoul National University, Seoul, Korea}
\affil[2]{Department of Statistics, Inha University, Incheon, Korea}
\affil[3]{Center for Dementia Research, Nathan S. Kline Institute for Psychiatric Research, New York, USA}
\affil[4]{Department of Psychiatry, NYU Grossman School of Medicine, New York, USA}

\begin{document}

\maketitle 

\begingroup
  \renewcommand{\thefootnote}{}
  \footnotetext{$^*$These authors contributed equally to this work.}
  \footnotetext{Corresponding author: Johan Lim, \texttt{johanlim@snu.ac.kr}}
\endgroup

\setcounter{footnote}{0}

\begin{abstract}
\noindent Testing the equality of two high-dimensional mean vectors is a fundamental problem in multivariate analysis. While the classical Hotelling's $T^2$ test is optimal in low-dimensional settings, it fails when the dimension $p$ is comparable to or exceeds the sample size $n$. Several extensions, including the Diagonal Likelihood Ratio Test (DLRT), have been proposed under the working independence assumption among variables. However, such an assumption can lead to a substantial loss of power when correlations are present. In this paper, we propose a new test, the Block Independent Likelihood Ratio Test (BILT), which generalizes DLRT by relaxing the working independence assumption to a block independence assumption. We establish its asymptotic normality of the null distirbution of the BILT statistic for `increasing $p$ with small $n$' under mild regularity conditions. We further analyze the asymptotic power of BILT under a local alternatives. Extensive simulation studies show that BILT maintains Type~I error control and achieves substantially higher power than DLRT across a wide range of covariance structures. An application to the Alzheimer’s Disease Neuroimaging Initiative (ADNI) dataset further demonstrates the application of  BILT to testing mean differences between two matrix-variate populations. 

\end{abstract}

\section{Introduction} 
In this paper, we consider the multivariate two-sample mean testing problem, where the two populations share a common covariance matrix. Specifically, let
\begin{equation}\label{eqn:random_sample}
\begin{aligned}
    \bX_i &= (X_{i1}, X_{i2}, \ldots, X_{ip})^\top \overset{\mathrm{i.i.d.}}{\sim}
    N_p(\bmu_1, \bSigma), \quad i = 1, \ldots, n_1, \\
    \bY_i &= (Y_{i1}, Y_{i2}, \ldots, Y_{ip})^\top \overset{\mathrm{i.i.d.}}{\sim}
    N_p(\bmu_2, \bSigma), \quad i = 1, \ldots, n_2,
\end{aligned}
\end{equation}
where the two samples are independent, and $N_p(\bmu, \bSigma)$ denotes the $p$-dimensional normal distribution with mean vector $\bmu$ and covariance matrix $\bSigma$. Under this setting, our interest lies in testing the hypothesis
\begin{equation}\label{eqn:hypothesis}
    H_0 : \bmu_1 = \bmu_2
    \quad \text{vs.} \quad
    H_1 : \bmu_1 \neq \bmu_2 .
\end{equation}

A classical and widely used approach for testing \eqref{eqn:hypothesis} is based on Hotelling’s $T^2$ statistic \citep{hotelling1931}, defined as
\begin{equation}\label{eqn:T2}
    T^2 = \frac{n_1 n_2}{N}
    (\bar{\bX} - \bar{\bY})^\top
    \bS^{-1}
    (\bar{\bX} - \bar{\bY}),
\end{equation}
where $N = n_1 + n_2$, $\bar{\bX} = n_1^{-1}\sum_{i=1}^{n_1} \bX_i$ and $\bar{\bY} = n_2^{-1}\sum_{i=1}^{n_2} \bY_i$ are the sample mean vectors of the two groups, and $\bS$ denotes the pooled sample covariance matrix, given by
\begin{equation*}
    \bS = \frac{1}{N-2} \Bigg(
    \sum_{i=1}^{n_1} (\bX_i - \bar{\bX})(\bX_i - \bar{\bX})^\top
    + \sum_{i=1}^{n_2} (\bY_i - \bar{\bY})(\bY_i - \bar{\bY})^\top
    \Bigg).
\end{equation*}
Under the null hypothesis $H_0$, the statistic in \eqref{eqn:T2} has the exact finite-sample distribution
\begin{equation*}
    \frac{N - p - 1}{(N - 2)p} T^2 \sim F_{p,\,N-p-1},
\end{equation*}
where $F_{\mathrm{df}_1, \mathrm{df}_2}$ denotes the $F$ distribution with numerator and denominator degrees of freedom $\mathrm{df}_1$ and $\mathrm{df}_2$, respectively. This distributional result provides a basis for testing the hypothesis in \eqref{eqn:hypothesis}.

However, it is well known that when the dimension $p$ approaches $N$, the power of Hotelling’s $T^2$ test decreases considerably \citep{bai1996effect}. Moreover, when $p > N-2$, the pooled sample covariance matrix $\bS$ becomes singular, and the Hotelling’s $T^2$ statistic in \eqref{eqn:T2} is no longer well defined. To address these challenges arising from high dimensionality, several modifications of the statistic in \eqref{eqn:T2} have been proposed in the literature. One line of work is based on regularization of the sample covariance matrix. For example, \citet{li2020adaptable} introduced regularized versions of Hotelling’s $T^2$ test by replacing the pooled sample covariance matrix $\bS$ in \eqref{eqn:T2} with a linear shrinkage estimator $\bS + \lambda \bI_p$, where $\bI_p$ denotes the $p$-dimensional identity matrix and $\lambda > 0$ is a regularization parameter controlling the amount of shrinkage. As another line of work, \citet{wu2006multivariate} and \citet{srivastava2008test} proposed replacing the pooled sample covariance matrix $\bS$ with its diagonal counterpart.\footnote{Although our focus is on the two-sample problem, related developments for the one-sample mean testing problem have also been studied; see, for example, \citet{chen2011regularized} for regularized method and \citet{srivastava2009test} for diagonal-based approach, among others.}

Beyond regularization and diagonal modifications of Hotelling’s $T^2$ statistic, several other lines of research have been developed. These include methods based on $L_2$-type statistics \citep{bai1996effect, chen2010two} for dense signal settings and maximum-type statistics \citep{tony2014two} for sparse signal settings; however, such methods are typically powerful only under specific alternatives, making it difficult to select a powerful test in practice. To address this limitation, \citet{xu2016adaptive} proposed an adaptive sum-of-powers test. Moreover, random projection-based methods \citep{lopes2011more, srivastava2016raptt}, optimal projection-based methods \citep{huang2015projection, liu2024projection}, and likelihood ratio-based methods \citep{hu2019diagonal} have been proposed. We refer the reader to \citet{huang2022overview} and the references therein for a comprehensive review of high-dimensional mean testing.

Among these approaches, our proposed method is inspired by the likelihood ratio based testing procedure of \citet{hu2019diagonal}. Specifically, \citet{hu2019diagonal} proposed a testing procedure, referred to as the Diagonal Likelihood Ratio Test (DLRT), which  constructs its test statistic under a working independence assumption, meaning that features are treated as mutually independent when forming the statistic, regardless of the true underlying dependence structure. They stablished the asymptotic normality of the null distribution of the DLRT statistic for `increasing $p$ with small $n$' under mild regularity conditions. DLRT takes account for the depedence among variables through the variance of the DLRT statistic, indirectly. It is one of a few procedures that are asymptotically valid in the large-$p$, small-$n$ regime.

In this paper, we  propose the Block Independence Likelihood Ratio Test (BILT), which generalizes the working independence assumption of \citet{hu2019diagonal} to a working block independence assumption. Unlike DLRT, BILT directly incorporates dependence information into the test statistic and thereby improves the power without sacrificing asymptotic validity. When the block size is set to one, BILT reduces to DLRT as a special case. The main contributions of this paper can be summarized as follows:

\begin{itemize}
    \item The BILT statistic is constructed as a likelihood ratio under the working block independence assumption, which, in this view, can be referred as a composite likelihood ratio test \citep{varin2005note, huang2020composite}. 

    \item We establish the asymptotic normality of the null distribution as $p$ increases, under mild regularity conditions. Here, the block-wise elementary statistics of BILT are quadratic forms following a scaled F-distribution, which differs from DLRT where the corresponding statistics follow a t-distribution; this distinction makes the extension to BILT non-trivial. We further derive the asymptotic power of BILT under local alternatives.
     
    \item Extensive numerical studies demonstrate that BILT (with block size $b > 1$) yields higher power than DLRT, except when the true covariance structure is diagonal, where  the working independence assumption of DLRT holds. However, even for this case, BILT (with block size $b > 1$) performs comparably to DLRT. Comparisons with several existing methods further show that BILT with a small block size (e.g. $b = 2$) consistently achieves superior performance.

    \item Finally, BILT naturally extends to testing mean differences in matrix-variate data. Given a sample of $\ell \times m$ matrices with $\ell < m$, each matrix can be vectorized into a $p = \ell \cdot m$-dimensional vector, and BILT can then be applied with a block size chosen as a multiple of $\ell$. A key advantage of this approach is that the choice of block size allows one to explicitly incorporate the dependency structure inherent in the matrix-valued observations: a block size of $\ell$ captures the row-wise dependency, while a larger block size additionally accounts for column-wise dependency. This flexibility enables BILT to achieve higher power than methods that ignore the dependency structure, such as DLRT. We illustrate this through an analysis of the Alzheimer's Disease Neuroimaging Initiative (ADNI) dataset.

\end{itemize}

\noindent \textbf{Outline.} The remainder of the paper is organized as follows. In Section~\ref{sec_2}, we briefly review the main ideas and theoretical results of the DLRT proposed by \citet{hu2019diagonal}. In Section~\ref{sec_3}, we introduce our proposed method, referred to as BILT, which generalizes the DLRT by incorporating block-wise dependence among variables into the construction of the test statistic. We derive the corresponding test statistic, establish its asymptotic normality, and provide the asymptotic power of the test for increasing $p$. Section~\ref{sec_4} presents extensive simulation studies to evaluate the finite-sample performance of BILT under various covariance structures and to compare its size and power to those of DLRT and other existing methods. In Section~\ref{sec_5}, we apply the proposed method to the ADNI dataset to show its applicability to testing the mean differences in matrix variate imaging data. Finally, Section~\ref{sec_6} concludes the paper with a summary and a discussion of potential directions for future research.   

\section{Preliminary} \label{sec_2}
In this section, we briefly review the DLRT proposed by \citet{hu2019diagonal} for the two-sample mean testing problem. We revisit the framework introduced in \eqref{eqn:random_sample}, where independent random samples are drawn from two independent normal populations, and consider testing the hypothesis in \eqref{eqn:hypothesis}. 

To facilitate the construction of the test statistic, \citet{hu2019diagonal} introduced a \emph{working covariance matrix}
\begin{equation}\label{eqn:w_cov_DLRT}
    \bSigma^{\mathrm{w}}_{\mathrm{DLRT}} = \mathrm{diag}(\sigma_{11}, \ldots, \sigma_{pp}) \in \mathbb{R}^{p \times p},    
\end{equation}
where $\sigma_{jj} = (\bSigma)_{jj}$ denotes the $j$th diagonal element of the true covariance matrix $\bSigma$. The superscript `$\mathrm{w}$' indicates that $\bSigma^{\mathrm{w}}_{\mathrm{DLRT}}$ is not the true covariance matrix but a working approximation. Treating $\bSigma^{\mathrm{w}}_{\mathrm{DLRT}}$ as if it were the true covariance matrix $\bSigma$ amounts to conducting the analysis under a \emph{working independence assumption}. Under this assumption, the joint likelihood function for the parameters $\bmu_1$, $\bmu_2$, and $\bSigma^{\mathrm{w}}_{\mathrm{DLRT}}$ is given by
\begin{equation*}
    L(\bmu_1, \bmu_2, \bSigma^{\mathrm{w}}_{\mathrm{DLRT}})
    = \prod_{j = 1}^p (2\pi)^{- \frac{N}{2}}
    \sigma_{jj}^{-\frac{N}{2}}
    \exp \left\{
    -\frac{1}{2\sigma_{jj}}
    \left(
    \sum_{i = 1}^{n_1} (X_{ij} - \mu_{1j})^2
    + \sum_{i = 1}^{n_2} (Y_{ij} - \mu_{2j})^2
    \right)
    \right\},
\end{equation*}
where $\mu_{1j}$ and $\mu_{2j}$ denote the $j$th components of the mean vectors $\bmu_1$ and $\bmu_2$, respectively. The DLRT statistic $T_{\mathrm{DLRT}}$ is defined as minus twice the log-likelihood ratio,
\begin{align*}
    T_{\mathrm{DLRT}} \coloneqq -2 \log(\Lambda_N)
    =
    -2 \log
    \left(
    \frac{
    \underset{\bmu_1, \bmu_2, \bSigma^{\mathrm{w}}_{\mathrm{DLRT}}}{\max}
    L(\bmu_1, \bmu_2, \bSigma^{\mathrm{w}}_{\mathrm{DLRT}} \mid H_0)
    }{
    \underset{\bmu_1, \bmu_2, \bSigma^{\mathrm{w}}_{\mathrm{DLRT}}}{\max}
    L(\bmu_1, \bmu_2, \bSigma^{\mathrm{w}}_{\mathrm{DLRT}})
    } 
    \right)
    = N \sum_{j = 1}^p \log\left(1 + \frac{t_{N,j}^2}{N - 2}\right),
\end{align*}
where
\begin{equation*}
    t_{N,j}^2
    = \frac{n_1 n_2}{N}
    \frac{(\bar{X}_j - \bar{Y}_j)^2}{s_{j,\mathrm{pool}}^2},
\end{equation*}
with $\bar{X}_j$ and $\bar{Y}_j$ denoting the sample means of the $j$th variable for the two groups, and $s_{j,\mathrm{pool}}^2$ the corresponding pooled sample variance.

For simplicity, let $V_{N,j} = N \log\!\left(1 + t_{N,j}^2/(N - 2)\right)$ for $j = 1, \ldots, p$. Then, the DLRT statistic can be expressed as $T_{\mathrm{DLRT}} = \sum_{j = 1}^p V_{N,j}$. Since the components $\{V_{N,j}\}_{j=1}^p$ are generally dependent, additional assumptions are required to establish the asymptotic behavior of $T_{\mathrm{DLRT}}$. \citet{hu2019diagonal} imposed the following regularity conditions:
\begin{itemize}
    \item[(C1)] 
    The stationary sequence $\{V_{N,j}\}$ satisfies the strong mixing condition such that
    \begin{equation*}
        \alpha_{\mathrm{DLRT}}(r) \coloneqq \sup_{1 \le k \le p-r}
        \alpha\!\left(\mathcal{F}_1^k, \mathcal{F}_{k+r}^p\right) \to 0 
        \quad \text{as } r \to \infty,    
    \end{equation*}
    where $\mathcal{F}_a^b = \sigma\{V_{N,j} : a \le j \le b\}$ and $\alpha(\mathcal{F}, \mathcal{G}) \coloneqq \sup_{A \in \mathcal{F},\, B \in \mathcal{G}} \lvert \mathbb{P}(A \cap B) - \mathbb{P}(A)\mathbb{P}(B) \rvert$ denotes the strong mixing coefficient between two $\sigma$-fields $\mathcal{F}$ and $\mathcal{G}$.
        
    \item[(C2)] 
    For some $\delta > 0$, $\sum_{r=1}^\infty \alpha_{\mathrm{DLRT}}(r)^{\delta/(2+\delta)} < \infty$, and for any $k \ge 0$, the following limit exists:
    \begin{equation*}
        \gamma_{\mathrm{DLRT}}(k)
        := \lim_{p \to \infty} \frac{1}{p-k}
        \sum_{j=1}^{p-k} \mathrm{Cov}\!\left(V_{N,j}, V_{N,j+k}\right).
    \end{equation*}
\end{itemize}

\noindent To describe the limiting distribution of the DLRT statistic, let $\Gamma(x) = \int_0^\infty t^{x-1} e^{-t}\,dt$ denote the gamma function, $\Psi(x) = \Gamma'(x)/\Gamma(x)$ the digamma function, and define $D(x) = \Psi\!\left((x+1)/2\right) - \Psi(x/2)$. Theorem~\ref{thm1} establishes the asymptotic normality of the DLRT statistic.

\begin{theorem}[Theorem~3 in \citet{hu2019diagonal}]\label{thm1}
Assume that conditions (C1) and (C2) hold and that the sequence $\{V_{N,j}\}$ is stationary. Then, for any fixed $N \ge 4$ and under the null hypothesis,
\begin{equation*}
    \frac{T_{\mathrm{DLRT}} - p G_1}{\tau_{\mathrm{DLRT}} \sqrt{p}}
    \;\overset{d}{\to}\;
    N(0,1)
    \quad \text{as} \quad p \to \infty,
\end{equation*}
where $G_1 \coloneq \Expect[V_{N,j}] = N D(N-2)$, $G_2 \coloneqq \Expect[V_{N,j}^2] = N^2\{D^2(N-2) - 2D'(N-2)\}$ and $\tau_{\mathrm{DLRT}}^2 = G_2 - G_1^2 + 2 \sum_{k = 1}^\infty \gamma_{\mathrm{DLRT}}(k)$.
\end{theorem}

In addition to the null distribution, \citet{hu2019diagonal} also investigate the asymptotic power of the DLRT under a local alternatives specified by
\begin{equation}\label{eqn:local_alter_DLRT}
    \bmu_1 - \bmu_2
    = \sqrt{\frac{N}{n_1 n_2}}\, \bdelta,
\end{equation}
where $\bdelta = (\delta_1, \ldots, \delta_p)^\top$. Define $\bDelta = (\delta_1/\sqrt{\sigma_{11}}, \ldots, \delta_p/\sqrt{\sigma_{pp}})^\top$ and assume that its components are uniformly bounded such that
\begin{equation}\label{eqn:unif_bound_DLRT}
    \max_{1 \le j \le p} |\Delta_j| \le M_{\mathrm{DLRT}},
\end{equation}
where $M_{\mathrm{DLRT}}$ is a constant independent of $N$ and $p$. Under these conditions, Theorem~\ref{thm2} characterizes the asymptotic power of the DLRT.

\begin{theorem}[Theorem~4 in \cite{hu2019diagonal}]\label{thm2}
    Suppose that $p$ diverges with $N$ at a rate satisfying $p = o(N^2)$. If the sequence $\{V_{N,j}\}$ is stationary and satisfies conditions (C1) and (C2), then under the local alternative in \eqref{eqn:local_alter_DLRT} and the uniform boundedness condition in \eqref{eqn:unif_bound_DLRT}, the asymptotic power of the level-$q$ test is given by
    \begin{equation*}
        \beta(T_{\mathrm{DLRT}}) = 1 - \Phi\left(z_q - \frac{\bDelta^\top \bDelta/\sqrt{p}}{\tau_{\mathrm{DLRT}}}\right)
        \quad \text{as} \quad (N,p) \to \infty,
    \end{equation*}
    where $z_q$ denotes the upper $q$-quantile of the standard normal distribution. 
\end{theorem}

DLRT is one of the few procedures that are asymptotically valid in the large-$p$, small-$n$ regime. However, DLRT cannot incorporate dependency information into the construction of the test statistic due to its working independence assumption, which leads to power loss when the true covariance matrix $\boldsymbol{\Sigma}$ is not diagonal. To address this limitation, we propose a method called BILT that incorporates the dependency structure into the construction of the test statistic, thereby improving power without sacrificing asymptotic validity.

\section{Method} \label{sec_3}
\subsection{BILT statistic} \label{sub_sec_3.1}
In this section, we present the Block Independent Likelihood ratio Test (BILT) for two-sample mean testing by extending the idea of DLRT to the more generous setting. 

Suppose that the $p$ components of $\bX_i \in \mathbb{R}^p$ are partitioned into $K$ groups of size $p_k$, so that $\bX_i = (\bX_{i1}^\top, \ldots, \bX_{iK}^\top)^\top$, where $\bX_{ik} \in \mathbb{R}^{p_k}$ with $\sum_{k=1}^K p_k = p$. The vector $\bY_i$ is partitioned analogously as $\bY_i = (\bY_{i1}^\top, \ldots, \bY_{iK}^\top)^\top$. Let $\bmu_{1k} = \Expect[\bX_{ik}]$, $\bmu_{2k} = \Expect[\bY_{ik}]$, and let $\bSigma_{kk}$ denote the common covariance matrix of $\bX_{ik}$ and $\bY_{ik}$. We relax the working independence assumption of \citet{hu2019diagonal} to a \emph{working block independence assumption}. In other words, we generalize the working covariance matrix $\bSigma_{\mathrm{DLRT}}^{\mathrm{w}}$ in \eqref{eqn:w_cov_DLRT} to a block diagonal matrix $\bSigma_{\mathrm{BILT}}^{\mathrm{w}}$, defined as
\begin{equation*}
    \bSigma_{\mathrm{BILT}}^{\mathrm{w}}
    =
    \begin{bmatrix}
        \bSigma_{11} & \bzero          & \cdots & \bzero \\
        \bzero       & \bSigma_{22}    & \cdots & \bzero \\
        \vdots       & \vdots          & \ddots & \vdots \\
        \bzero       & \bzero          & \cdots & \bSigma_{KK}
    \end{bmatrix}
    \in \mathbb{R}^{p \times p}.
\end{equation*}
Under the working block independence assumption, the joint likelihood function for $\bmu_1$, $\bmu_2$, and $\bSigma_{\mathrm{BILT}}^{\mathrm{w}}$ is given by
\begin{equation*}
\begin{split}
L(\bmu_1, \bmu_2, \bSigma_{\mathrm{BILT}}^{\mathrm{w}})
&=
\prod_{k = 1}^K
(2\pi)^{-\frac{N}{2}}
\det(\bSigma_{kk})^{-\frac{N}{2}}
\exp\!\Bigg\{
-\frac{1}{2}
\Bigg[
\sum_{i = 1}^{n_1}
(\bX_{ik} - \bmu_{1k})^\top
\bSigma_{kk}^{-1}
(\bX_{ik} - \bmu_{1k})\\
&\phantom{\prod_{k = 1}^K
(2\pi)^{-\frac{N}{2}}
\det(\bSigma_{kk})^{-\frac{N}{2}}
\exp\!\Bigg\{
-\frac{1}{2}} \quad 
+
\sum_{i = 1}^{n_2}
(\bY_{ik} - \bmu_{2k})^\top
\bSigma_{kk}^{-1}
(\bY_{ik} - \bmu_{2k})
\Bigg]
\Bigg\}.
\end{split}
\end{equation*}
The BILT statistic $T_{\mathrm{BLRT}}$ is then define as minus twice the log-likelihood ratio,
\begin{equation*}
    T_{\mathrm{\mathrm{BILT}}} \coloneqq -2 \log(\Lambda_N)
    =
    -2 \log
    \left(
    \frac{
    \underset{\bmu_1, \bmu_2, \bSigma^{\mathrm{w}}_{\mathrm{BILT}}}{\max}
    L(\bmu_1, \bmu_2, \bSigma^{\mathrm{w}}_{\mathrm{BILT}} \mid H_0)
    }{
    \underset{\bmu_1, \bmu_2, \bSigma^{\mathrm{w}}_{\mathrm{BILT}}}{\max}
    L(\bmu_1, \bmu_2, \bSigma^{\mathrm{w}}_{\mathrm{BILT}})
    } 
    \right)
    = N \sum_{k = 1}^K \log\left(1 + \frac{A_{N,k}}{N - 2}\right),
\end{equation*}
where
\begin{equation*}
    A_{N,k}
    = \frac{n_1 n_2}{N} (\bar{\bX}_{k} - \bar{\bY}_k)^\top \bS_k^{-1} (\bar{\bX}_k - \bar{\bY}_k),
\end{equation*}
with $\bar{\bX}_k$ and $\bar{\bY}_k$ denoting the sample mean vectors corresponding to the $k$th block in the two groups, and $\bS_k$ the corresponding pooled sample covariance matrix.

\subsection{Asymptotic normality} \label{sub_sec_3.2}

Let $U_{N,k} = N \log\!\left(1 + A_{N,k}/(N-2)\right)$. We begin by deriving the mean and variance of $U_{N,k}$. Following arguments analogous to those used in \citet{hu2019diagonal}, we establish the Lemma~\ref{lemma1}. The proof of Lemma~\ref{lemma1} is provided in Appendix~\ref{Appendix_A}.
\begin{lemma}\label{lemma1}
    For any $N \geq \underset{1\leq k \leq K}{\max}\,p_k + 2$,
    \begin{align*}
    \Expect[U_{N,k}] &= N D_{p_k}(N-2), \\
    \mathrm{Var}(U_{N,k}) &= -2 N^2 D_{p_k}'(N-2),
    \end{align*}
    where $D_s(x) \coloneqq \sum_{j = x - s + 1}^{x} D(j) = \Psi\!\left((x+1)/2\right) - \Psi\!\left((x - s + 1)/2\right)$.
    Moreover,
    \[
    \Expect[U_{N,k}] \to p_k
    \quad \text{and} \quad
    \mathrm{Var}(U_{N,k}) \to 2p_k \quad \text{as}\quad N \to \infty.
    \]
\end{lemma}

Now, we turn to show the asymptotic normality of $T_{\mathrm{BILT}} = \sum_{k = 1}^K U_{N,k}$. Similar to the DLRT statistic, the sequence $\{U_{N,k}\}_{k=1}^K$ is not mutually independent, and thus the classical central limit theorem does not apply directly. To address this issue, we consider the following regularity condition:
\begin{itemize}
    \item[(C1$^\prime$)]
    The stationary sequence $\{U_{N,k}\}$ satisfies the strong mixing condition such that
    \begin{equation*}
        \alpha_{\mathrm{BILT}}(r) \coloneqq \sup_{1 \le l \le K-r}
        \alpha\!\left(\mathcal{F}_1^l, \mathcal{F}_{l+r}^K\right) \to 0 
        \quad \text{as } r \to \infty.
    \end{equation*}

    \item[(C2$^\prime$)]
    For some $\delta > 0$, $\sum_{r=1}^\infty \alpha_{\mathrm{BILT}}(r)^{\delta/(2+\delta)} < \infty$, and for any $l \geq 0$, the following limit exists:
    \begin{equation*}
        \gamma_{\mathrm{BILT}}(l) := \lim_{K \to \infty} \frac{1}{K-l}
        \sum_{k=1}^{K-l} \mathrm{Cov}\!\left(U_{N,k}, U_{N,k+l}\right).
    \end{equation*}
\end{itemize}

\noindent Under these regularity conditions, we establish the asymptotic normality of the BILT statistic as stated in Theorem~\ref{thm3}, whose proof is given in Appendix~\ref{Appendix_B}.

\begin{theorem} \label{thm3}
    Assume that the block size $p_k$ is fixed for some positive integer $b$ for all $k$. Furthermore, assume that the condition (C1$^\prime$) and (C2$^\prime$) hold and that the sequence $\{U_{N,k}\}$ is stationary.  Then, for any fixed $N \geq b + 3$ and under the null hypothesis,
    \begin{equation*}
        \frac{T_{\mathrm{BILT}} - KND_b(N-2)}{\tau_{\mathrm{BILT}}\sqrt{K}} \;\overset{d}{\to}\; N(0,1) \quad \text{as} \quad p \to \infty,
    \end{equation*}
    where $\tau_{\mathrm{BILT}}^2 = -2N^2D^\prime_{b}(N-2) + 2\sum_{l = 1}^\infty \gamma_{\mathrm{BILT}}(l)$.
\end{theorem} 
 
\begin{remark}
    BILT may be viewed as an extension of DLRT in that, when each block is restricted to size one, it reduces to DLRT. For block sizes greater than one, BILT incorporates dependence information into the construction of the test statistic while preserving asymptotic validity. Consequently, when the true covariance matrix $\bSigma$ deviates from the identity matrix $\bI_p$, BILT can achieve higher power than DLRT. Empirical results further indicate that even with a block size of two, BILT yields noticeable power gains; see Section~\ref{sec_4} for details.
\end{remark}

\begin{remark}\label{remark_2}
    As the block size $b$ increases, the block working independence assumption in our model becomes weaker, leading to the lower model bias. However, simultaneously, the effective sample size $K$, which determines the convergence rate of $T_{\mathrm{BILT}}=\sum_{k=1}^K U_{N,k}$, decreases, resulting in increased model variance. Therefore, it is necessary to find an optimal $b$ which maximizes test performance. To this end, traditional parameter selection methods, such as cross-validation, may be employed. In Section \ref{sec_4}, we examine the variation in test performance according to block size in greater detail through simulation studies.
\end{remark}
\subsection{Consistent estimator for $\tau_{\mathrm{BILT}}^2$}\label{sub_sec_3.3}
Suppose that the block size $p_k$ is fixed at $b \ge 1$ for all $k$ and that a consistent estimator of $\tau_{\mathrm{BILT}}^2$, denoted by $\hat{\tau}_{\mathrm{BILT}}^2$, is available. Then, by Slutsky’s theorem together with Theorem~\ref{thm3}, the standardized BILT statistic $\tilde{Z}_{\mathrm{BILT}}$ converges in distribution to the standard normal distribution as $p \to \infty$:
\begin{equation*}
    \tilde{Z}_{\mathrm{BILT}}
    := \frac{T_{\mathrm{BILT}} - K N D_b (N-2)}{\hat{\tau}_{\mathrm{BILT}} \sqrt{K}}
    \;\overset{d}{\to}\; N(0,1)
    \quad \text{as} \quad p \to \infty .
\end{equation*}
Accordingly, we can construct an asymptotically valid test that rejects the null hypothesis whenever $|\tilde{Z}_{\mathrm{BILT}}| \ge z_{q/2}$, where $z_q$ denotes the upper $q$-quantile of the standard normal distribution and $q$ is the nominal significance level.

To implement the above decision rule, a consistent estimator of $\tau_{\mathrm{BILT}}^2$ is required. Let $\bar{U} = \frac{1}{K} \sum_{k=1}^K U_{N,k}$ and define the sample lag-$l$ covariance by
\begin{equation*}
    \widehat{\gamma}(l) 
    = \frac{1}{K-l} \sum_{k=1}^{K-l} 
    \bigl(U_{N,k} - \bar{U}\bigr)
    \bigl(U_{N,k+l} - \bar{U}\bigr).
\end{equation*}
Based on these quantities, we consider a kernel-based long-run variance estimator of $\tau_{\mathrm{BILT}}^{2}$ defined as
\begin{equation}\label{eqn:kernel_based_est}
    \hat{\tau}_{\mathrm{BILT}}^2 
    = -2N^2 D_b^{\prime}(N-2) 
    + 2 \sum_{l=1}^{L} 
    w\!\left(\frac{l}{L}\right)
    \widehat{\gamma}(l),
\end{equation}
where $w(\cdot)$ is a kernel function and $L$ is a bandwidth parameter. Various kernel choices have been studied in the literature, including the Parzen kernel and the quadratic spectral kernel; see \citet{andrews1991heteroskedasticity} for a comprehensive discussion. In particular, taking $w(x)=\mathbb{I}(|x| \leq 1)$ yields the sample-based estimator.

Under mild regularity conditions, kernel-based long-run variance estimators of the form \eqref{eqn:kernel_based_est} are known to be consistent; see, for example, \citet{newey1986simple}, \citet{andrews1991heteroskedasticity}, and \citet{jansson2002consistent}. Accordingly, we adopt the estimator in \eqref{eqn:kernel_based_est} as a consistent estimator of $\tau_{\mathrm{BILT}}^2$.

\subsection{Asymptotic Power of the BILT}\label{sub_sec_3.4}
In this section, we analyze the asymptotic power of the BILT under a high-dimensional regime in which both $N$ and $p$ diverge to infinity. As in \cite{hu2019diagonal}, we consider a local alternatives defined by
\begin{equation}\label{eqn:local_alter_BILT}
    \bmu_1 - \bmu_2 = \sqrt{\frac{N}{n_1n_2}} \bdelta,
\end{equation}
where $\bdelta = (\bdelta_1^\top, \ldots, \bdelta_K^\top)^\top \in \mathbb{R}^p$. Let 
$\bDelta = (\Delta_1, \ldots, \Delta_K)^\top 
= (\bdelta_1^\top \bSigma_{11}^{-1}\bdelta_1, \ldots, \bdelta_K^\top \bSigma_{KK}^{-1}\bdelta_K)^\top \in \mathbb{R}^K$, 
and assume that its components are uniformly bounded in the sense that
\begin{equation}\label{eqn:unif_bound_BILT}
    \underset{1\leq k \leq K}{\max} \Delta_k \leq M_{\mathrm{BILT}}, 
\end{equation}
where $M_{\mathrm{BILT}}$ is a constant independent of $N$ and $p$. Under these conditions, Theorem~\ref{thm4} characterizes the asymptotic power of the BILT.

\begin{theorem}\label{thm4}
Assume that the block size $p_k$ is fixed to some positive integer $b$ for all $k$ and that the number of blocks $K$ increases at such a rate that $K = o(N^2)$. If the sequence $\{U_{N, k}\}$ is stationary and satisfies conditions (C1$^\prime$) and (C2$^\prime$), then under the local alternative \eqref{eqn:local_alter_BILT} and condition \eqref{eqn:unif_bound_BILT}, the asymptotic power of the level-$q$ test is given by
\begin{equation*}
    \beta(T_{\mathrm{BILT}}) 
    = 1 - \Phi\left(
        z_q - \frac{\bDelta^\top \bDelta / \sqrt{K}}{\tau_{\mathrm{BILT}}}
      \right) 
    \quad \text{as}  \quad  (N, p) \to \infty.
\end{equation*}
\end{theorem}

\noindent Theorem~\ref{thm4} implies that $\beta(T_{\mathrm{BILT}}) \to 1$ if $\sqrt{K} = o(\bDelta^\top \bDelta)$, whereas $\beta(T_{\mathrm{BILT}}) \to q$ if $\bDelta^\top \bDelta = o(\sqrt{K})$. The proof is given in Appendix~\ref{Appendix_C}.

\section{Numerical Study} \label{sec_4}
In this section, we conduct numerical experiments to assess the performance of BILT and compare it with existing methods. We generate independent samples $\bX_i \overset{\mathrm{i.i.d.}}{\sim} N_p(\bmu_1, \bSigma)$, $i = 1, \ldots, n_1$, and $\bY_i \overset{\mathrm{i.i.d.}}{\sim} N_p(\bmu_2, \bSigma)$, $i = 1, \ldots, n_2$. Throughout, we set $\bmu_1 = \bzero$. Unless otherwise specified, each component of $\bmu_2$ is independently drawn from $\{-\delta/\sqrt{p},\, \delta/\sqrt{p}\}$ with probability $1/2$ each, where $\delta > 0$ controls the signal magnitude. We consider four covariance structures for $\bSigma$, defined as follows:
\begin{itemize}
    \item \textbf{Independent} (Ind): $\bSigma = \bI_p$.

    \item \textbf{Autoregressive} (AR): $(\bSigma)_{ij} = \rho^{|i-j|}$ for $1 \le i,j \le p$.

    \item \textbf{Block-diagonal} (BD): $\bSigma$ is block-diagonal with common block size $4$. 
    Specifically, let $\lfloor x \rfloor$ denote the greatest integer less than or equal to $x$. Then
    \begin{equation*}
    (\bSigma)_{ij}=
    \begin{cases}
    1, & i=j,\\
    \rho, & i\neq j \ \text{and}\ 
    \left\lfloor\frac{i-1}{4}\right\rfloor=
    \left\lfloor\frac{j-1}{4}\right\rfloor,\\
    0, & \text{otherwise}.
    \end{cases}
    \end{equation*}

    \item \textbf{Band} (BAND): $\bSigma$ is a banded covariance matrix with bandwidth $4$, that is,
    \begin{equation*}
    (\bSigma)_{ij}=
    \begin{cases}
    1, & i=j,\\
    \rho, & 0<|i-j|\le 4,\\
    0, & |i-j|>4.
    \end{cases}
    \end{equation*}
\end{itemize}
\noindent For the AR and BD structures, we consider $\rho \in \{0.3,\, 0.6\}$, denoted by AR\_0.3, AR\_0.6 and BD\_0.3, BD\_0.6, respectively. For the BAND structure, we fix $\rho = 0.3$, denoted by BAND\_0.3. 

Under each configuration, we compute the standardized BILT statistic
\begin{equation*}
    \tilde{Z}_{\mathrm{BILT}}
    =
    \frac{
    T_{\mathrm{BILT}} - K N D_b (N-2)
    }{
    \hat{\tau}_{\mathrm{BILT}}\sqrt{K}
    },
\end{equation*}
where $\hat{\tau}_{\mathrm{BILT}}^2$ is a kernel-based estimator of $\tau_{\mathrm{BILT}}^2$. Throughout the simulation studies and real data analyses, we employ the Parzen kernel with fixed bandwidth $L = 5$. The null hypothesis $H_0:\bmu_1 = \bmu_2$ is rejected at nominal level $q$ whenever $|\tilde{Z}_{\mathrm{BILT}}| > z_{q/2}$. The empirical type~I error rate and power are estimated as the proportion of rejections over $r$ Monte Carlo replications. Unless otherwise stated, we set $q = 0.05$ and $r = 3{,}000$.

\subsection{Asymptotic Normality of the Standardized BILT Statistic}\label{sub_sec_4.1}
We first investigate the asymptotic normality of the standardized BILT statistic $\tilde{Z}_{\mathrm{BILT}}$ under the null hypothesis $\delta = 0$. We set $n_1 = n_2 = 50$ and $p = 1{,}400$, and examine the null distribution under two covariance structures, AR\_0.6 and BD\_0.6. BILT is implemented with block sizes $b \in \{1, 2, 5, 10\}$. 
For each configuration, we compute $\tilde{Z}_{\mathrm{BILT}}$ based on $5{,}000$ independent replications and display the resulting histograms together with the standard normal density. The results under AR\_0.6 are presented in Figure~\ref{fig:1}. The results under BD\_0.6, which show similar patterns, are reported in Appendix~\ref{Appendix_D}.

Several patterns can be observed in Figure~\ref{fig:1}: When the block size is $b = 1$, corresponding to DLRT, the histogram of $\tilde{Z}_{\mathrm{BILT}}$ deviates noticeably from the standard normal density, particularly around zero and in the tails. As the block size increases, the histogram aligns more closely with the standard normal density. Even when $b = 2$, a clear improvement is observed. These findings provide empirical support for the asymptotic normality of BILT established in Theorem~\ref{thm3} and suggest that block sizes larger than $1$ may lead to a more accurate finite-sample approximation to the limiting normal distribution.

\begin{figure}[!htb]
    \centering
    \includegraphics[width=\linewidth]{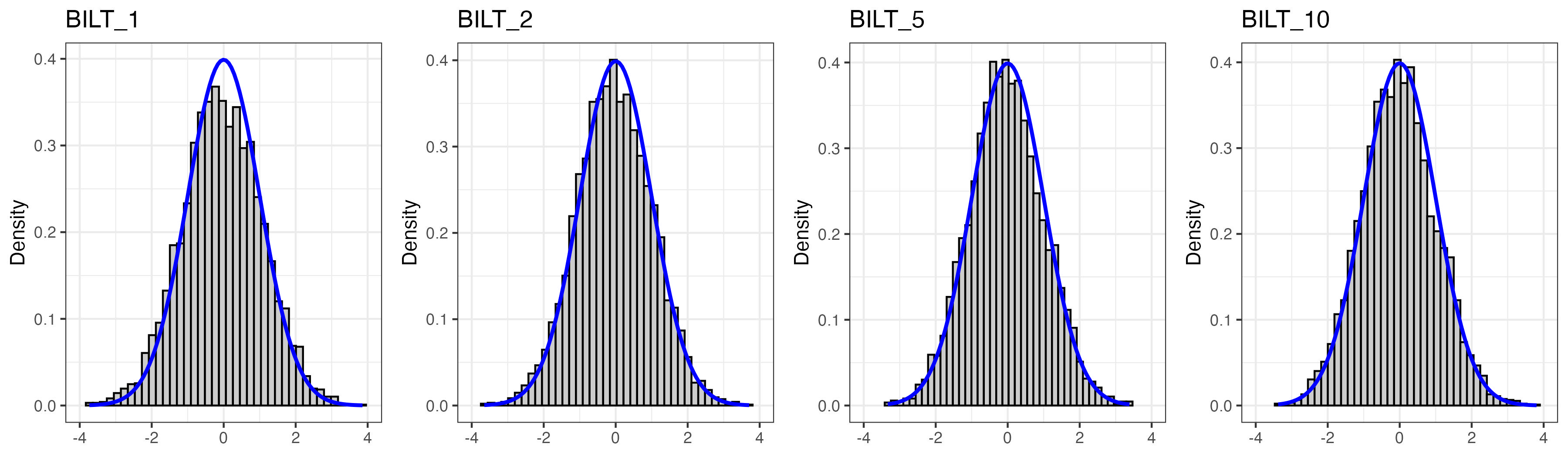}
    \caption{Histogram of the standardized BILT statistics based on 5{,}000 replications, overlaid with the standard normal density, under the AR\_0.6 covariance structure. From left to right, the panels present the results for BILT with block sizes $b = 1, 2, 5,$ and $10$, respectively.}
    \label{fig:1}
\end{figure}

\subsection{Type I Error Control}\label{sub_sec_4.2}
Next, we examine the type~I error of BILT as the dimension $p \in \{200, 400, \ldots, 1{,}400\}$ increases to assess its asymptotic validity. We consider two sample sizes, $n_1 = n_2 = 50$ and $n_1 = n_2 = 100$, and compare the type~I error rates across six covariance structures and block sizes $b \in \{1, 2, 5, 10\}$. Figure~\ref{fig:2} presents the results for $n_1 = n_2 = 50$. The results for $n_1 = n_2 = 100$, which show similar trends, are provided in Appendix~\ref{Appendix_D}.

Figure~\ref{fig:2} shows that when $b = 1$, corresponding to DLRT, BILT controls the type~I error close to the nominal level of $0.05$ under relatively weak dependence structures, such as IND, AR\_0.3, and BD\_0.3, across all dimensions $p$. However, as the dependence becomes stronger---namely under BAND\_0.3, AR\_0.6, and BD\_0.6---the type~I error deviates noticeably from the nominal level, regardless of $p$. 
In contrast, for $b > 1$, the type~I error tends to approach the nominal level as $p$ increases, irrespective of the covariance structure considered. This improvement is particularly pronounced under stronger dependence structures such as BAND\_0.3, AR\_0.6, and BD\_0.6. Notably, BILT with $b = 2$ consistently achieves superior or at least comparable type~I error control relative to $b = 1$ across all simulation settings. These findings suggest that $b = 2$ serves as a reasonable default choice when the underlying dependence structure is unknown.
When the dimension $p$ is small, however, BILT with block sizes $b \in \{5, 10\}$ fails to adequately control the type~I error. This behavior appears to stem from instability in the variance estimator of $T_{\mathrm{BILT}}$ when $p$ is small relative to the block size, as discussed in Remark~\ref{remark_2}.

\begin{figure}[!htb]
    \centering
    \includegraphics[width=\linewidth]{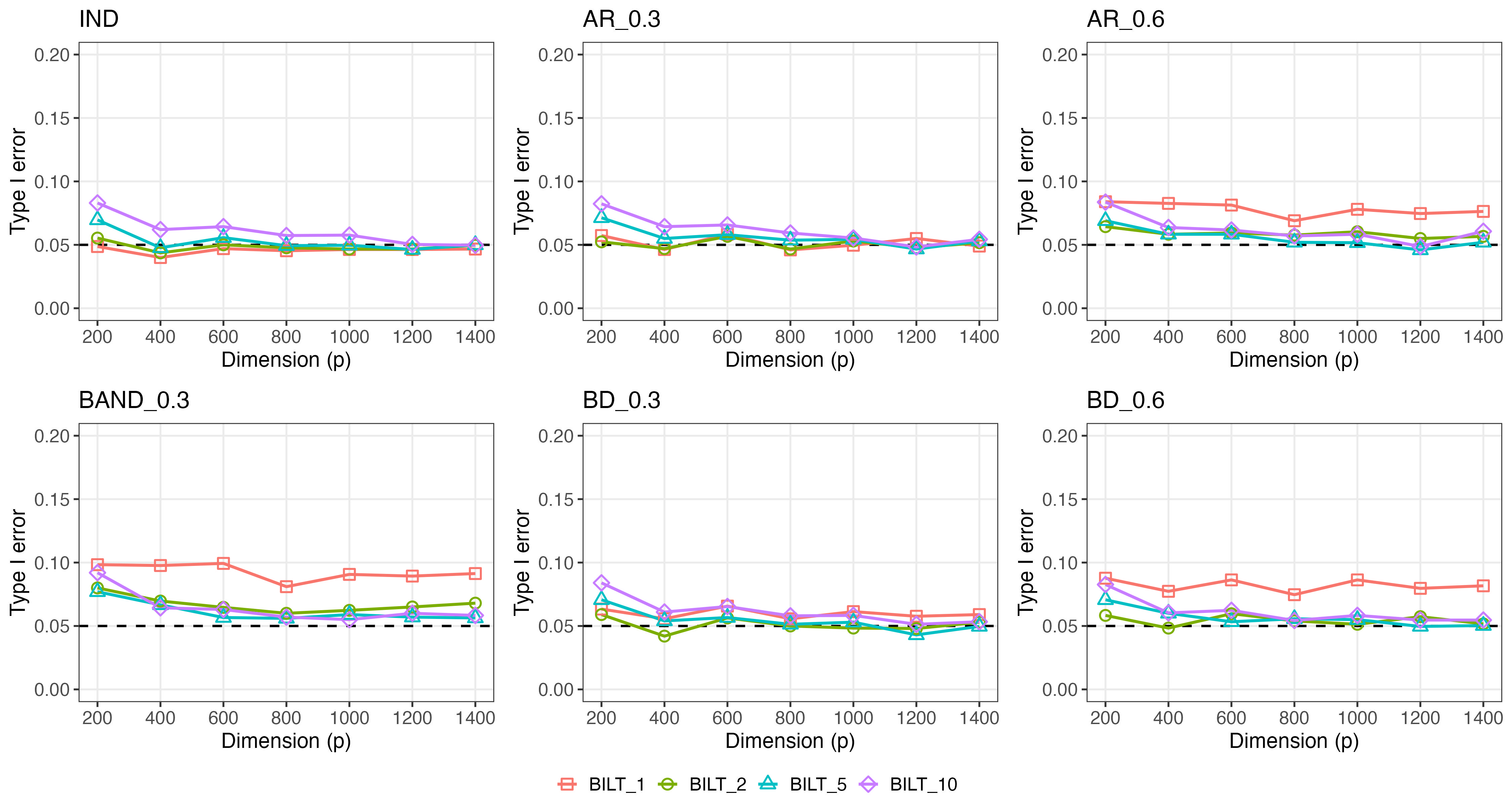}
    \caption{Type~I error of BILT with $n_1 = n_2 = 50$. Each panel corresponds to a covariance structure. Within each panel, results for block sizes $b \in \{1, 2, 5, 10\}$ are displayed. The x-axis represents the dimension $p$, and the y-axis represents the Type~I error obtained from 3{,}000 replications. The horizontal dashed line indicates the nominal Type~I error level of $0.05$.}
    \label{fig:2}
\end{figure}

\subsection{Power Curve Against Signal Magnitude}\label{sub_sec_4.3} 
Now, we examine the power of BILT as the signal magnitude $\delta/\sqrt{p}$ increases. With $p = 500$ fixed and $n_1 = n_2 \in \{50, 100\}$, we vary $\delta$ to generate different signal magnitudes and compare empirical power across six covariance structures and block sizes $b \in \{1, 2, 5, 10\}$.
The results for $n_1 = n_2 = 50$ are presented in Figure~\ref{fig:3}. The results for $n_1 = n_2 = 100$, which show similar trends, are reported in Appendix~\ref{Appendix_D}, 

From Figure~\ref{fig:3}, several observations emerge. 
Overall, power increases monotonically with the signal magnitude, which is both natural and expected. 
When the common covariance structure is IND, the working independence assumption is exactly satisfied, and BILT with block size $b = 1$ corresponds to the oracle case. Nevertheless, the four methods exhibit nearly identical power curves, indicating that increasing the block size $b$ under independence does not adversely affect the power of BILT. 
In contrast, for the AR, BD, and BAND covariance structures, increasing the block size consistently improves power, with the improvement becoming more pronounced under stronger dependence structures, namely BAND\_0.3, AR\_0.6, and BD\_0.6. The power gain relative to DLRT can be attributed to relaxing the working independence assumption toward a blockwise dependence structure, thereby more accurately capturing within-block correlations in constructing the test statistic. 
Notably, this improvement remains evident even for the relatively small block size $b = 2$. Taken together with the type~I error results in Figure~\ref{fig:2}, these findings further reinforce the earlier conclusion that BILT with $b = 2$ provides a sensible default choice when the underlying covariance structure is unknown.

\begin{figure}[!htb]
    \centering
    \includegraphics[width=\linewidth]{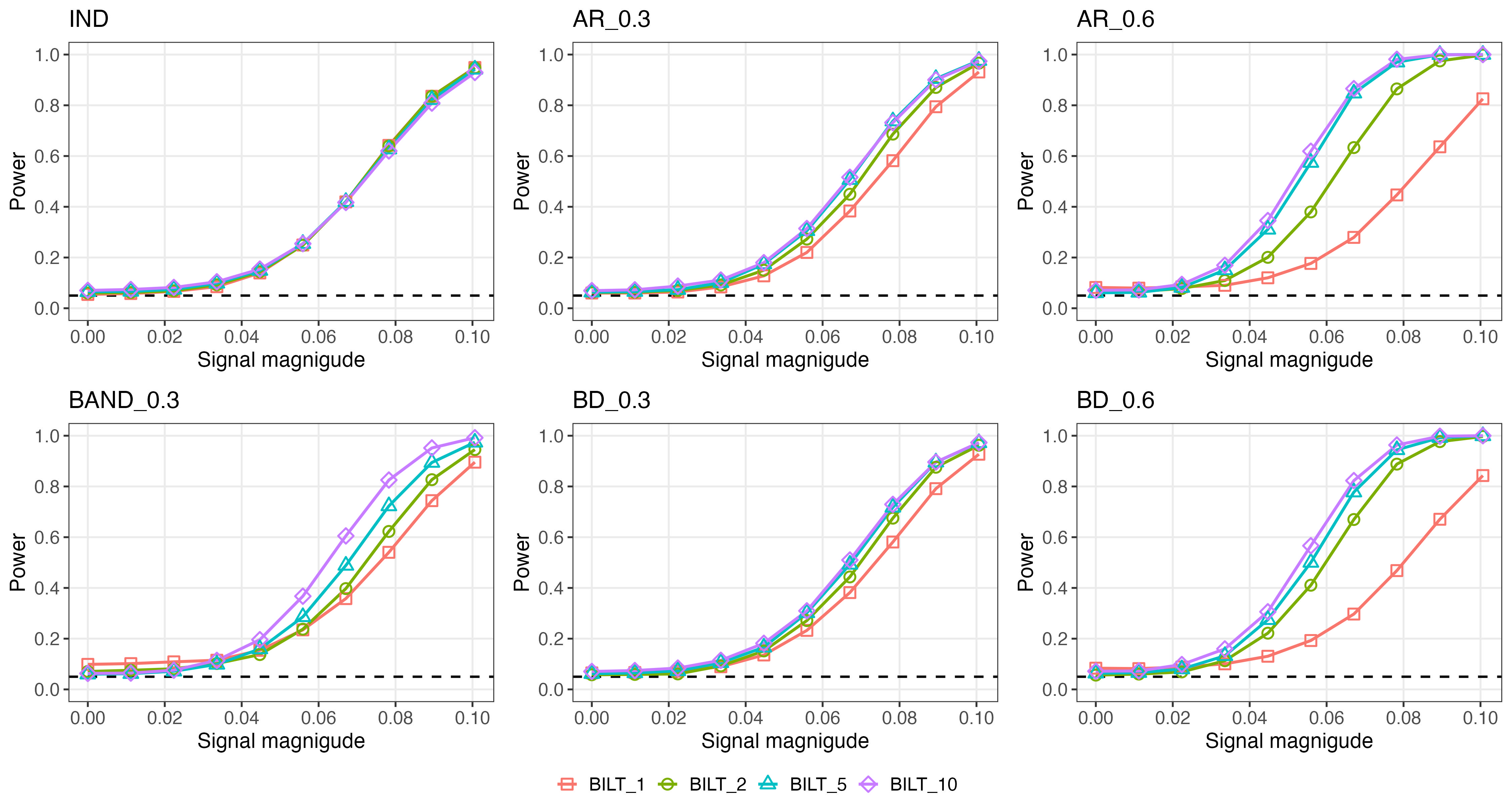}
    \caption{Power of BILT with $n_1 = n_2 = 50$ and $p = 500$. Each panel corresponds to a covariance structure. Within each panel, results for block sizes $b \in \{1, 2, 5, 10\}$ are displayed. The x-axis represents the signal magnitude $\delta/\sqrt{p}$, and the y-axis represents the power obtained from 3{,}000 replications. The horizontal dashed line indicates the nominal Type~I error level of $0.05$.}
    \label{fig:3}
\end{figure}

\subsection{Power Curve Against Non-Null Proportion}\label{sub_sec_4.4}
We further examine power of BILT as the non-null proportion increases. To this end, we randomly select a subset of $\{1,2,\ldots,p\}$ according to a prespecified non-null proportion. The components $\mu_{2j}$ corresponding to the selected indices are independently assigned values in $\{-\delta/\sqrt{p},\, \delta/\sqrt{p}\}$ with probability $1/2$ each, while the remaining components are set to zero. Consequently, as the non-null proportion increases, the signal becomes progressively denser. We fix $p = 1{,}000$ and $\delta = 3$, and consider sample sizes $n_1 = n_2 \in \{50, 100\}$. We then compare empirical power across six common covariance structures and block sizes $b \in \{1, 2, 5, 10\}$. The results for $n_1 = n_2 = 50$ are presented in Figure~\ref{fig:4}; those for $n_1 = n_2 = 100$ show similar patterns and are provided in Appendix~\ref{Appendix_D}.

Figure~\ref{fig:4} displays patterns similar to those observed in Figure~\ref{fig:3}. 
Overall, power increases monotonically as the non-null proportion increases, reflecting the fact that a denser signal strengthens the aggregated mean difference and thereby enhances detectability. 
When the common covariance structure is IND, all four methods exhibit nearly identical power curves across all non-null proportions. 
In contrast, under the remaining covariance structures, power generally increases with the block size $b$, and this effect becomes particularly pronounced under stronger dependence structures, such as BAND\_0.3, AR\_0.6, and BD\_0.6. The improvement in power is still evident even for $b = 2$.

\begin{figure}[!htb]
    \centering
    \includegraphics[width=\linewidth]{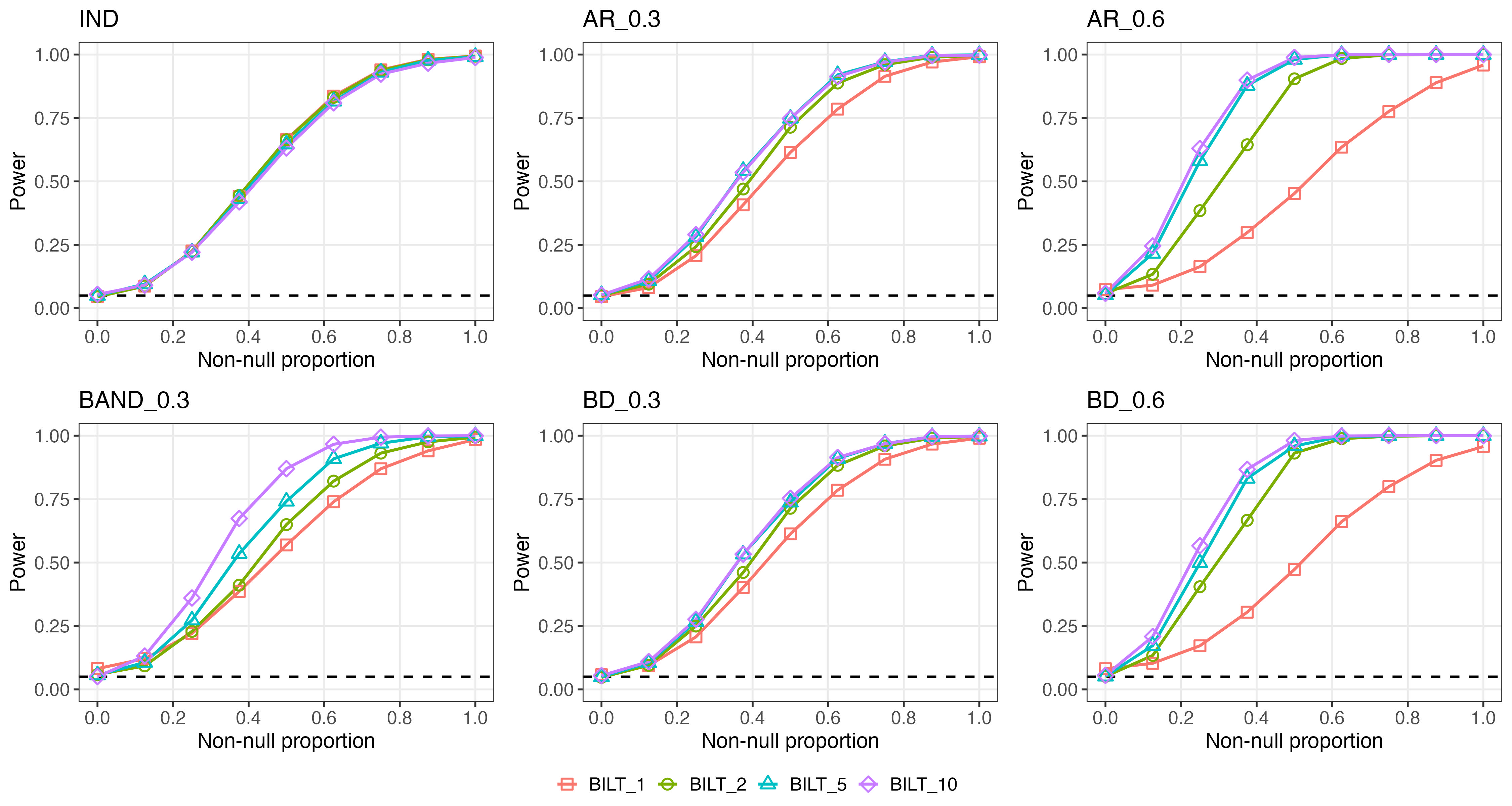}
    \caption{Power of BILT with $n_1 = n_2 = 50$ and $p = 1{,}000$. Each panel corresponds to a covariance structure. Within each panel, results for block sizes $b \in \{1, 2, 5, 10\}$ are displayed. The x-axis represents the non-null proportion, and the y-axis represents the power obtained from 3{,}000 replications. The horizontal dashed line indicates the nominal Type~I error level of $0.05$.}
    \label{fig:4}
\end{figure}

\subsection{Comparison with Other Methods}\label{sub_sec_4.5}
Finally, we compare the Type I error and power of BILT with block size $b = 2$, denoted by BILT\_2, with that of six existing tests: the BS test \citep{bai1996effect}, the SD test \citep{srivastava2008test}, the CQ test \citep{chen2010two}, the ARHT test \citep{li2020adaptable}, the aSPU test \citep{xu2016adaptive} and the DLRT \citep{hu2019diagonal}.
The BS, SD, CQ, and aSPU tests are implemented via the \texttt{R} package \texttt{highmean}, whereas the ARHT test is implemented via the \texttt{R} package \texttt{ARHT}.

We assess the power of each method under various non-null proportion, following the simulation design in Section~\ref{sub_sec_4.4}. We fix $n_1 = n_2 = 50$ and $p = 500$, and set $\delta$ to $2.25$. In addition, we introduce heterogeneity in the variances. Specifically, the diagonal elements of $\bSigma$ are generated as $\sigma_{ii} \sim \chi^2_5 / 5$ for $i = 1, \ldots, p$, rather than being fixed at 1. For each configuration, the experiment is repeated 1{,}000 times, and the empirical Type~I error and power are computed for all tests. The results are summarized in Figure~\ref{fig:5}.

We observe the following patterns. When the covariance structure is IND, all methods asymptotically control the Type~I error at the nominal level, and SD exhibits the highest power across all non-null proportions. The power curve of BILT\_2 is slightly lower than that of SD, but the difference is negligible. For the other covariance structures, BILT\_2 maintains control of the Type~I error while achieving power that is comparable to, and often higher than, that of the competing methods. Its advantage becomes more pronounced when the dependence structure is relatively strong, such as in the AR\_0.6 and BD\_0.6 cases.

\begin{figure}[!htb]
    \centering
    \includegraphics[width=\linewidth]{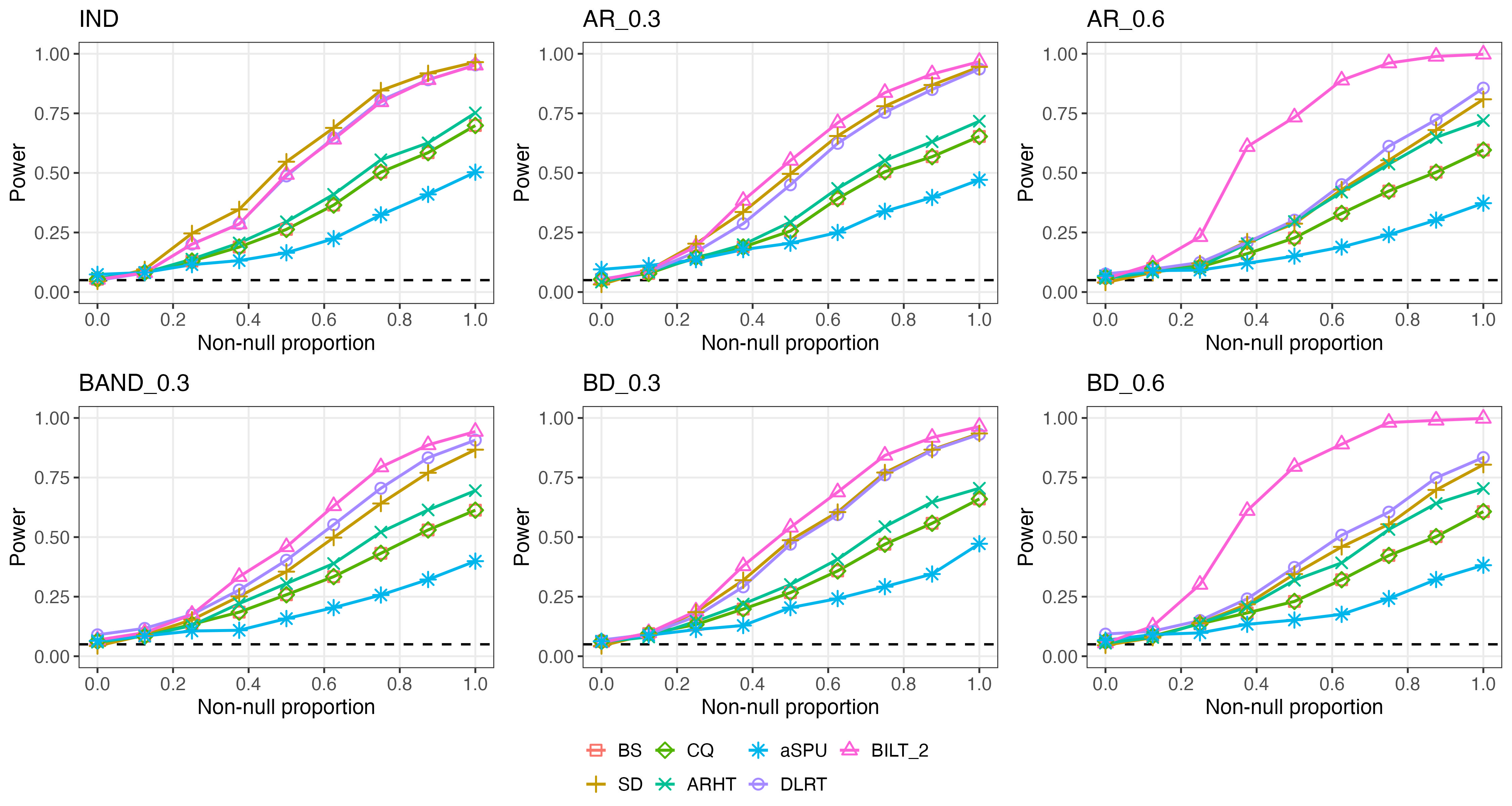}
    \caption{Power comparison of BS, SD, CQ, ARHT, aSPU, DLRT, and BILT\_2 with $n_1 = n_2 = 50$ and $p = 500$. Each panel corresponds to a covariance structure. Within each panel, results for the seven tests are displayed. The x-axis represents the non-null proportion, and the y-axis represents the power based on 1{,}000 replications. The horizontal dashed line indicates the nominal Type~I error level of $0.05$.}
    \label{fig:5}
\end{figure}

\section{Real Data Analysis} \label{sec_5}

In this section, we illustrate the application of BILT to testing the mean difference between two populations with matrix-valued observations using the Alzheimer's Disease Neuroimaging Initiative (ADNI) dataset, and compare its performance with existing methods.

\subsection{Data Description}
 The ADNI dataset includes MRI-based measurements of the corpus callosum (CC) for 135 subjects (95 males, 40 females) initially diagnosed with mild cognitive impairment (MCI). Subjects are divided into two groups: 74 converters (MCI-C), who progressed to Alzheimer’s disease (AD) within three years of their MCI diagnosis, and 61 non-converters (MCI-NC), who remained free of AD during the follow-up period. For further details on data collection and preprocessing, see \citet{lee2016predicting}.

Although the original ADNI dataset contains a variety of predictors and imaging metrics, we focus exclusively on the CC thickness profiles measured over a two-year period. The left panel of Figure~\ref{fig:6} displays a brain MRI image with the CC, while the right panel shows the CC partitioned into 99 equally spaced intervals, each yielding a single thickness measurement. This results in a vector of 99 thickness measurements per subject, and since measurements are collected over two time points, the data for each subject can be organized into a $2 \times 99$ matrix, where the two rows correspond to the two time points and the 99 columns correspond to spatial locations along the CC. It is important to emphasize that the two rows are temporally dependent, while the 99 columns exhibit spatial dependence. For analytical convenience, we discard the 99th measurement, reducing each subject's data to a $2 \times 98$ matrix. All subsequent analyses are based on these reduced matrices. 

\begin{figure}[!htb]
    \centering
    \includegraphics[width=0.9\linewidth]{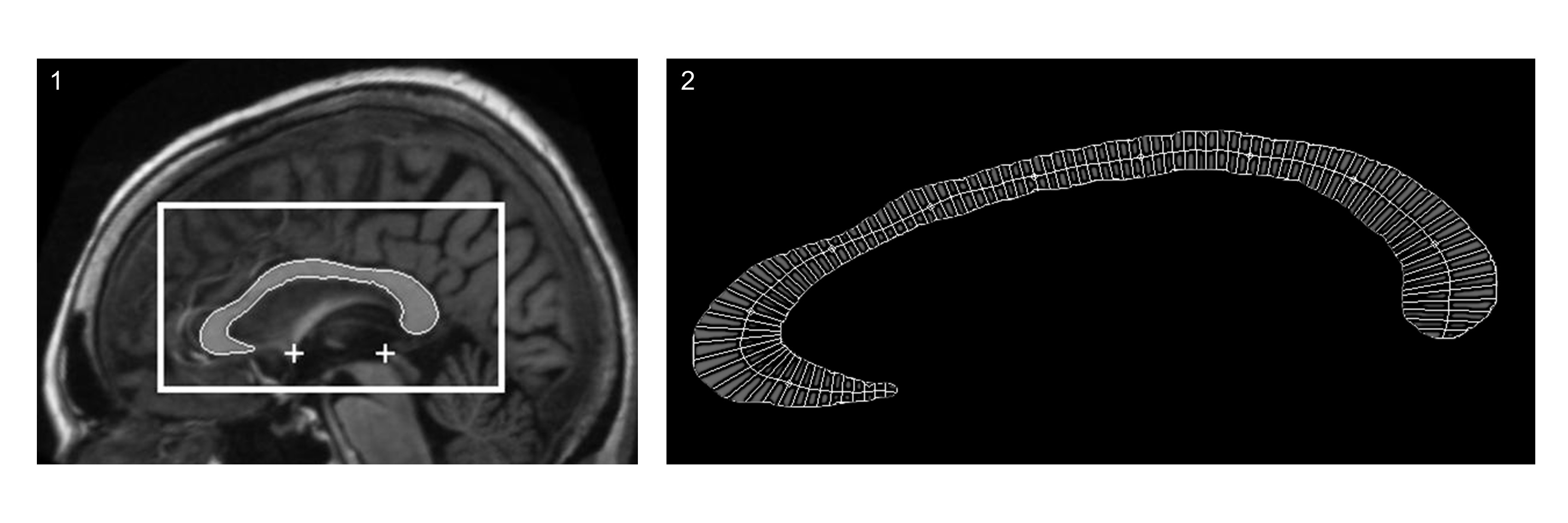}
    \caption{(1) A brain MRI image with the corpus callosum (CC) highlighted by its outline. (2) The measurement of CC thickness profile overlaid on a segmented cross-section of the CC.}
    \label{fig:6}
\end{figure}

\subsection{Experimental Setup}
With these matrix-valued observations, our goal is to investigate differences in CC thickness profiles across the following three groupings:
\begin{itemize}
    \item \textbf{Conversion}: MCI-C vs.\ MCI-NC
    \item \textbf{Gender}: Male vs.\ Female
    \item \textbf{Interaction}: (MCI-C, Male) $\cup$ (MCI-NC, Female) vs.\ (MCI-NC, Male) $\cup$ (MCI-C, Female)
\end{itemize}
For example, the \textbf{Conversion} grouping compares CC thickness profiles between subjects who converted to AD (MCI-C) and those who did not (MCI-NC). The third grouping is specifically designed to capture the interaction effect between AD conversion and gender.

For the analysis, each subject's $2 \times 98$ matrix is vectorized into a $196$-dimensional vector by stacking the two measurements at each spatial location consecutively. Based on these vectors, we apply BILT\_2, BILT\_4, DLRT, SD, and ARHT across the three groupings. Specifically, BILT\_2 is designed to account for temporal dependency, while BILT\_4 is designed to account for both temporal and spatial dependencies; Figure~\ref{fig:7} provides a visual illustration of the dependency structures captured by each method. In BILT\_2, each block consists of a $2 \times 1$ subvector containing the two temporal measurements at a single spatial location, thereby capturing temporal dependency. In BILT\_4, each block consists of a $2 \times 2$ submatrix grouping two consecutive spatial locations across both time points, thereby capturing both temporal and spatial dependencies simultaneously. DLRT serves as a baseline that ignores any dependency structure, and SD and ARHT are included as representative benchmarks, as they exhibited relatively strong power in the simulation studies of Section~\ref{sec_4}.

\begin{figure}
    \centering
    \includegraphics[width=\linewidth]{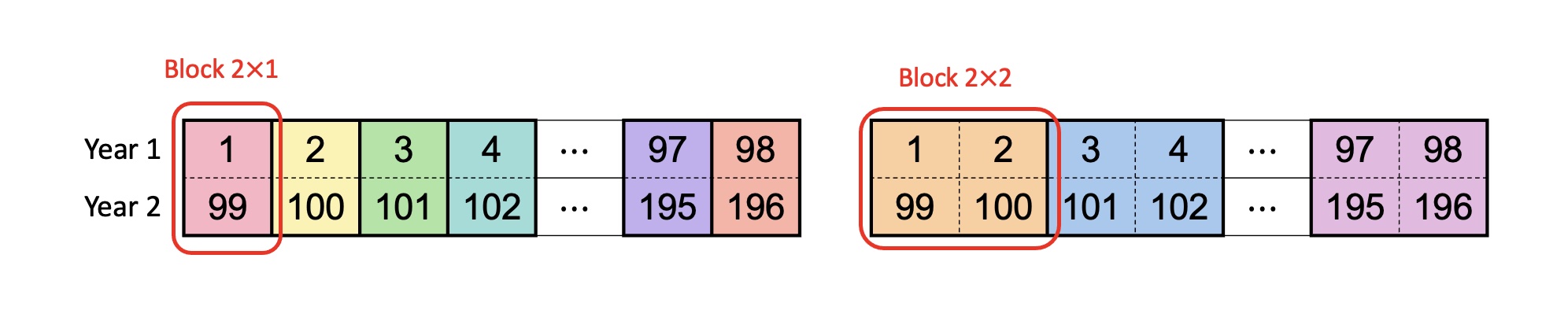}
    \caption{Illustration of the block structures used in BILT\_2 and BILT\_4. In BILT\_2 (left), each block has size $2 \times 1$, grouping the two temporal measurements at a single spatial location (same color). In BILT\_4 (right), each block has size $2 \times 2$, grouping two consecutive spatial locations across both time points (same color), thereby additionally capturing spatial dependency.}
    \label{fig:7}
\end{figure}

\subsection{Result}

Table \ref{tab:realdata} summarizes the $p$-values for each experiment. First, we observe that the DLRT and BILT methods demonstrate noticeably higher statistical power than SD and ARHT.
When comparing DLRT and BILT, both methods produced highly significant results in cases where the group difference is substantial, such as the MCI-C versus MCI-NC. 
However, regarding the interaction effect, both BILT\_2 and BILT\_4 demonstrated considerably higher significance than DLRT. This suggests that accommodating the temporal dependency through block structure enhances testing performance. Furthermore, for the gender comparison, only BILT\_4 yielded a significant result while the other methods showed weak significance. This implies two key points: first, the spatial and temporal dependencies inherent in the CC thickness profile data were successfully captured BILT\_4; and second, despite the numerical study results showing the best performance for BILT\_2, using a larger block size may have practical utility depending on the underlying data structure.
In particular, the spatial dependency structure of the CC thickness data is expected to resemble an autoregressive (AR) structure, where correlations decay over distance. Thus, it is intuitive that BILT naturally outperforms DLRT by capturing these local dependencies, rather than assuming working independence.

\begin{table}[H]
\centering
\caption{Comparison of p-values for each hypothesis among various testing methods. Each row represents a tested hypothesis setting, and each column indicates the testing method applied.}
\begin{tabular}{cccccc}
\hline
                 & BILT\_2  & BILT\_4  & DLRT     & SD    & ARHT  \\ \hline
Conversion & $<$0.001 & $<$0.001 & $<$0.001 & 0.212 & 0.108 \\ \hline
Gender  & 0.262    & 0.040    & 0.197    & 0.541 & 0.075 \\ \hline
Interaction      & 0.026    & 0.050    & 0.130    & 0.540 & 0.817 \\ \hline
\end{tabular}
\label{tab:realdata}
\end{table}



\section{Conclusion} \label{sec_6}
In this study, we proposed BILT as a generalization of DLRT for high-dimensional two-sample mean testing. By incorporating a flexible block-independence structure, BILT accommodates moderate within-block correlations commonly observed in high-dimensional data while retaining computational feasibility.
We established the asymptotic normality of the BILT statistic under mild regularity conditions. In addition, we analyzed the asymptotic power of BILT under local alternatives. Extensive simulation studies demonstrated that BILT achieves asymptotic Type~I error control and yields substantial power gains over DLRT, particularly under correlated structures that deviate from the independence covariance model. 
Furthermore, the application to the Alzheimer’s Disease Neuroimaging Initiative (ADNI) dataset highlighted the practical value of BILT. 
Overall, the proposed BILT framework offers a promising direction for high-dimensional inference where balancing dependence modeling and statistical efficiency is essential.

There are several directions for future research. These include the development of data-driven methods for selecting the block size $b$, extensions to settings with unequal covariance matrices, and adaptations to more general dependence structures, such as spatial or temporal correlations. In addition, it would be worthwhile to investigate the finite-sample performance of the variance estimator for $T_{\mathrm{BILT}}/\sqrt{K}$ under different choices of kernel functions $w(\cdot)$ and bandwidth parameters $L$. In particular, several natural questions arise in finite samples: which kernel function is optimal in maximizing power while maintaining asymptotic Type~I error control, and how should one select an optimal bandwidth parameter $L$ in a data-adaptive manner.

\section*{Acknowledgement}
The authors declare that there are no conflicts of interest.

\section*{Data Availability Statement}
Data used in preparation of this article were obtained from the Alzheimer's Disease Neuroimaging Initiative (ADNI) database (\url{adni.loni.usc.edu}). As such, the investigators within the ADNI contributed to the design and implementation of ADNI and/or provided data but did not participate in analysis or writing of this report. A complete listing of ADNI investigators can be found at: \url{http://adni.loni.usc.edu/wp-content/uploads/how_to_apply/ADNI_Acknowledgement_List}.

\bibliographystyle{chicago}
\bibliography{ref}

@article{hu2019diagonal,
  title={Diagonal likelihood ratio test for equality of mean vectors in high-dimensional data},
  author={Hu, Zongliang and Tong, Tiejun and Genton, Marc G},
  journal={Biometrics},
  volume={75},
  number={1},
  pages={256--267},
  year={2019}
}

@article{huang2022overview,
  title={An overview of tests on high-dimensional means},
  author={Huang, Yuan and Li, Changcheng and Li, Runze and Yang, Songshan},
  journal={Journal of Multivariate Analysis},
  volume={188},
  pages={104813},
  year={2022}
}

@article{bai1996effect,
  title={Effect of high dimension: by an example of a two sample problem},
  author={Bai, Zhidong and Saranadasa, Hewa},
  journal={Statistica Sinica},
  volume={6},
  number={2},
  pages={311--329},
  year={1996}
}

@article{chen2010two,
title={A two-sample test for high-dimensional data with applications to gene-set testing},
author={Song Xi Chen and Ying-Li Qin},
journal={The Annals of Statistics},
volume={38},
number={2},
pages={808--835},
year={2010}
}

@article{srivastava2008test,
  title={A test for the mean vector with fewer observations than the dimension},
  author={Srivastava, Muni S and Du, Meng},
  journal={Journal of Multivariate Analysis},
  volume={99},
  number={3},
  pages={386--402},
  year={2008}
}

@article{tony2014two,
  title={Two-sample test of high dimensional means under dependence},
  author={Cai, Tony and Liu, Weidong and Xia, Yin},
  journal={Journal of the Royal Statistical Society Series B: Statistical Methodology},
  volume={76},
  number={2},
  pages={349--372},
  year={2014}
}

@article{xu2016adaptive,
  title={An adaptive two-sample test for high-dimensional means},
  author={Xu, Gongjun and Lin, Lifeng and Wei, Peng and Pan, Wei},
  journal={Biometrika},
  volume={103},
  number={3},
  pages={609--624},
  year={2016}
}

@article{chen2011regularized,
  title={A regularized {Hotelling}’s {T}$^2$ test for pathway analysis in proteomic studies},
  author={Chen, Lin S and Paul, Debashis and Prentice, Ross L and Wang, Pei},
  journal={Journal of the American Statistical Association},
  volume={106},
  number={496},
  pages={1345--1360},
  year={2011}
}

@article{hotelling1931,
title = {The Generalization of {Student}'s Ratio},
author = {Harold Hotelling},
journal = {The Annals of Mathematical Statistics},
volume = {2},
number = {3},
pages = {360 -- 378},
year = {1931}
}

@article{lee2016predicting,
  title={Predicting progression from mild cognitive impairment to Alzheimer's disease using longitudinal callosal atrophy},
  author={Lee, Sang Han and Bachman, Alvin H and Yu, Donghyeon and Lim, Johan and Ardekani, Babak A and Alzheimer's Disease Neuroimaging Initiative and others},
  journal={Alzheimer's \& Dementia: Diagnosis, Assessment \& Disease Monitoring},
  volume={2},
  number = {1},
  pages={68--74},
  year={2016}
}

@article{andrews1991heteroskedasticity,
  title={Heteroskedasticity and autocorrelation consistent covariance matrix estimation},
  author={Andrews, Donald WK},
  journal={Econometrica},
  volume={59},
  number={3},
  pages={817--858},
  year={1991}
}

@article{wu2006multivariate,
  title={A multivariate two-sample mean test for small sample size and missing data},
  author={Wu, Yujun and Genton, Marc G and Stefanski, Leonard A},
  journal={Biometrics},
  volume={62},
  number={3},
  pages={877--885},
  year={2006}
}

@article{srivastava2009test,
  title={A test for the mean vector with fewer observations than the dimension under non-normality},
  author={Srivastava, Muni S},
  journal={Journal of Multivariate Analysis},
  volume={100},
  number={3},
  pages={518--532},
  year={2009}
}

@inproceedings{lopes2011more,
  title={A more powerful two-sample test in high dimensions using random projection},
  author={Lopes, Miles and Jacob, Laurent and Wainwright, Martin J},
  journal={Advances in Neural Information Processing Systems},
  volume={24},
  year={2011},
  booktitle = {Advances in Neural Information Processing Systems}
}

@article{liu2024projection,
  title={Projection test for mean vector in high dimensions},
  author={Liu, Wanjun and Yu, Xiufan and Zhong, Wei and Li, Runze},
  journal={Journal of the American Statistical Association},
  volume={119},
  number={545},
  pages={744--756},
  year={2024}
}

@phdthesis{huang2015projection,
  author={Huang, Yuan},
  title={Projection test for high-dimensional mean vectors with optimal direction},
  school={The Pennsylvania State University},
  year={2015} 
}

@article{srivastava2016raptt,
  title={{RAPTT}: An exact two-sample test in high dimensions using random projections},
  author={Srivastava, Radhendushka and Li, Ping and Ruppert, David},
  journal={Journal of Computational and Graphical Statistics},
  volume={25},
  number={3},
  pages={954--970},
  year={2016}
}

@book{hall2014martingale,
  title={Martingale Limit Theory and Its Application},
  author={Hall, Peter and Heyde, Christopher C},
  year={2014},
  publisher={Academic Press},
  address={New York}
}

@article{newey1986simple,
  title={A simple, positive semi-definite, heteroskedasticity and autocorrelation consistent covariance matrix},
  author={Newey, Whitney K and West, Kenneth D},
  journal={Econometrica},
  volume={55},
  number={3},
  pages={703--708},
  year={1986}
}

@article{jansson2002consistent,
  title={Consistent covariance matrix estimation for linear processes},
  author={Jansson, Michael},
  journal={Econometric Theory},
  volume={18},
  number={6},
  pages={1449--1459},
  year={2002}
}

@article{li2020adaptable,
  title={An adaptable generalization of {Hotelling}'s {T}$^2$ test in high dimension},
  author={Li, Haoran and Aue, Alexander and Paul, Debashis and Peng, Jie and Wang, Pei},
  journal={The Annals of Statistics},
  volume={48},
  number={3},
  pages={1815--1847},
  year={2020}
}

@article{varin2005note,
  title={A note on composite likelihood inference and model selection},
  author={Varin, Cristiano and Vidoni, Paolo},
  journal={Biometrika},
  volume={92},
  number={3},
  pages={519--528},
  year={2005}
}

@article{huang2020composite,
  title={Composite likelihood inference under boundary conditions},
  author={Huang, Jing and Ning, Yang and Cai, Yi and Liang, Kung-Yee and Chen, Yong},
  journal={Statistica Sinica},
  volume={30},
  number={2},
  pages={1005--1025},
  year={2020}
}

\appendix
\numberwithin{equation}{section}
\numberwithin{table}{section}
\numberwithin{figure}{section}

\numberwithin{theorem}{section}
\numberwithin{lemma}{section}
\numberwithin{proposition}{section}
\numberwithin{corollary}{section}

\section{Proof of Lemma~\ref{lemma1}} \label{Appendix_A}
Lemma~\ref{lemmaA1} provides a useful results that facilitates the proof of Lemma~\ref{lemma1}.
\begin{lemma}\label{lemmaA1}
    For any $a, b > 0$, the following integral equalities hold:
    \begin{align*}
        \int_0^1 \log (z) \, z^{\frac{a}{2}-1} (1-z)^{\frac{b}{2}-1} \, dz
        &= B\!\left(\frac{a}{2}, \frac{b}{2}\right)
        \left\{
            \Psi\!\left(\frac{a}{2}\right)
            - \Psi\!\left(\frac{a+b}{2}\right)
        \right\}, \\[0.5em]
        \int_0^1 (\log z)^2 \, z^{\frac{a}{2}-1} (1-z)^{\frac{b}{2}-1} \, dz
        &= B\!\left(\frac{a}{2}, \frac{b}{2}\right)
        \Bigg[
            \left\{
                \Psi\!\left(\frac{a}{2}\right)
                - \Psi\!\left(\frac{a+b}{2}\right)
            \right\}^2
            + \left\{
                \Psi'\!\left(\frac{a}{2}\right)
                - \Psi'\!\left(\frac{a+b}{2}\right)
            \right\}
        \Bigg],
    \end{align*}
    where $B(a,b)$ denotes the beta function.
\end{lemma}
\begin{proof}
   By definition of the beta function, $B(a/2,b/2)$ admits the integral representation
    \begin{equation}\label{eqn:beta_func}
        B\!\left(\frac{a}{2}, \frac{b}{2}\right)
        = \int_0^1 z^{\frac{a}{2}-1} (1-z)^{\frac{b}{2}-1} \, dz.
    \end{equation}
    We first differentiate the left-hand side of \eqref{eqn:beta_func} with respect to $a$. Using the identity
    $B(a,b)=\Gamma(a)\Gamma(b)/\Gamma(a+b)$ together with $\Psi(x)=\Gamma'(x)/\Gamma(x)$ yields
    \begin{equation}\label{eqn:diff_LHS}
        \frac{\partial}{\partial a}
        B\!\left(\frac{a}{2}, \frac{b}{2}\right)
        = \frac{1}{2}
        B\!\left(\frac{a}{2}, \frac{b}{2}\right)
        \left\{
            \Psi\!\left(\frac{a}{2}\right)
            - \Psi\!\left(\frac{a+b}{2}\right)
        \right\}.
    \end{equation}
    On the other hand, differentiating the right-hand side of \eqref{eqn:beta_func} with respect to $a$ yields
    \begin{equation}\label{eqn:diff_RHS}
        \frac{\partial}{\partial a}
        \int_0^1 z^{\frac{a}{2}-1} (1-z)^{\frac{b}{2}-1} \, dz
        = \frac{1}{2} \int_0^1 \log(z)\, z^{\frac{a}{2}-1} (1-z)^{\frac{b}{2}-1} \, dz .
    \end{equation}
    Equating \eqref{eqn:diff_RHS} and \eqref{eqn:diff_LHS} establishes the first equality.
    
    For the second equality, we differentiate both sides of \eqref{eqn:beta_func} twice with respect to $a$. After straightforward calculations using $\Psi'(x)$ (the trigamma function) and the differentiation under the integral argument, the desired identity follows.
\end{proof}

\medskip
\begin{proof}[Proof of Lemma~\ref{lemma1}]
    Note that
    \begin{equation}\label{eqn:U_Nk}
        U_{N,k} \;\overset{d}{=}\; N \log\!\left(1 + \frac{p_k}{N-p_k-1} F_k\right),
        \qquad
        F_k \sim F(p_k,\, N-p_k-1).
    \end{equation}
    Thus, it suffices to evaluate $\Expect\{\log(1+\tfrac{a}{b}F)\}$ for $F \sim F(a,b)$. We have
    \begin{align*}
        \Expect\!\left[\log\!\left(1+\frac{a}{b}F\right)\right]
        &= \int_{0}^{\infty} \log\!\left(1+\frac{a}{b}x\right)
            \frac{(ax)^{a/2-1} b^{\,b/2}}{(ax+b)^{(a+b)/2}}
            \frac{a}{B(a/2,b/2)} \, dx \\[0.5em]
        &= B(a/2,b/2)^{-1}
            \int_{0}^{\infty} \log\!\left(\frac{ax+b}{b}\right)
            \left(\frac{ax}{ax+b}\right)^{a/2-1}
            \left(\frac{b}{ax+b}\right)^{b/2-1}
            \frac{ab}{(ax+b)^2} \, dx \\[0.5em]
        &\overset{(\star)}{=} -\, B(a/2,b/2)^{-1}
        \int_{0}^{1} \log y \, (1-y)^{a/2-1} y^{\,b/2-1} \, dy \\[0.5em]
        &\overset{(\star\star)}{=} \Psi\!\left(\frac{a+b}{2}\right) - \Psi\!\left(\frac{b}{2}\right).
    \end{align*}
    The equality $(\star)$ follows from the change of variables $y = \frac{b}{ax+b}$, and the equality $(\star\star)$ follows from Lemma~\ref{lemmaA1}.
    Consequently, by substituting $a=p_k$ and $b=N-p_k-1$, we obtain
    \begin{equation*}
        \Expect[U_{N,k}]
        = N \left\{
        \Psi\!\left(\frac{N-1}{2}\right)
        - \Psi\!\left(\frac{N-p_k-1}{2}\right)
        \right\}
        = N\, D_{p_k}(N-2).    
    \end{equation*}
    Similarly, we obtain
    \begin{align*}
        \mathbb{E}\!\left[\left\{\log\!\left(1 + \frac{a}{b}F\right)\right\}^2\right]
        &= B(a/2,b/2)^{-1}
        \int_{0}^{1} (\log y)^2 \, (1-y)^{a/2-1} y^{\,b/2-1} \, dy \\[0.5em]
        &=
        \Bigg[
            \left\{
                \Psi\!\left(\frac{b}{2}\right)
                - \Psi\!\left(\frac{a+b}{2}\right)
            \right\}^2
            + \left\{
                \Psi'\!\left(\frac{b}{2}\right)
                - \Psi'\!\left(\frac{a+b}{2}\right)
            \right\}
        \Bigg],
    \end{align*}
    where the last equality follows from Lemma~\ref{lemmaA1}.  
    Substituting $a=p_k$ and $b=N-p_k-1$ yields
    \begin{align*}
        \mathbb{E}[U_{N,k}^2]
        &= N^2 \Bigg[
            \left\{
                \Psi\!\left(\frac{N-1}{2}\right)
                - \Psi\!\left(\frac{N-p_k-1}{2}\right)
            \right\}^2
            - \left\{
                \Psi'\!\left(\frac{N-1}{2}\right)
                - \Psi'\!\left(\frac{N-p_k-1}{2}\right)
            \right\}
        \Bigg] \\
        &= N^2 \left\{ D_{p_k}^2(N-2) - 2 D_{p_k}'(N-2) \right\}.
    \end{align*}
    Consequently,
    \[
    \mathrm{Var}(U_{N,k})
    = \mathbb{E}[U_{N,k}^2] - \{\mathbb{E}[U_{N,k}]\}^2
    = -\,2 N^2 D_{p_k}'(N-2).
    \]

    We now study the asymptotic mean and variance of $U_{N,k}$.
    Since the identity in \eqref{eqn:U_Nk} hold, by the density of the $F$ distribution, the expectation of $U_{N,k}$ can be written as
    \begin{align*}
    \mathbb{E}[U_{N,k}]
    &= N \frac{p_k\,\Gamma\!\left(\frac{N-1}{2}\right)}
            {\Gamma\!\left(\frac{p_k}{2}\right)\Gamma\!\left(\frac{N-p_k-1}{2}\right)}
       \int_0^{\infty}
       \log\!\left(1+\frac{p_k}{N-p_k-1}x\right)
       \frac{(p_k x)^{\frac{p_k}{2}-1}(N-p_k-1)^{\frac{N-p_k-1}{2}}}
            {(p_k x+N-p_k-1)^{\frac{N-1}{2}}}
       \, dx \\[0.5em]
    &= \frac{p_k}{\Gamma\!\left(\frac{p_k}{2}\right)}
       \underbrace{\left[
       \frac{\Gamma\!\left(\frac{N-1}{2}\right)}
            {\Gamma\!\left(\frac{N-p_k-1}{2}\right)}
       (N-p_k-1)^{-\frac{p_k}{2}}
       \right]}_{\eqqcolon \textup{(I)}}\\
       &\qquad \times \underbrace{\int_0^{\infty}
       N \log\!\left(1+\frac{p_k}{N-p_k-1}x\right)
       \left(1+\frac{p_k x}{N-p_k-1}\right)^{-\frac{N-1}{2}}
       (p_k x)^{\frac{p_k}{2}-1}\, dx}_{\eqqcolon \textup{(II)}} .
    \end{align*}
    We analyze each term as $N \to \infty$.
    For term~\textup{(I)}, applying Stirling’s approximation, $\Gamma(x)=\sqrt{2\pi}\,x^{x-\frac12}e^{-x}\{1+O(x^{-1})\}$, yields
    \begin{equation}\label{eqn:I}
        \lim_{N\to\infty} \textup{(I)}
        = 2^{-\frac{p_k}{2}}.    
    \end{equation}
    For term~\textup{(II)}, an application of the dominated convergence theorem yields
    \begin{equation}\label{eqn:II}
        \lim_{N\to\infty} \textup{(II)} = \int_0^{\infty} (p_k x)^{\frac{p_k}{2}} e^{-\frac{p_k x}{2}} \, dx 
    \end{equation}
    Combining \eqref{eqn:I} and \eqref{eqn:II} yields
    \begin{align*}
    \lim_{N\to\infty} \mathbb{E}[U_{N,k}]
    &= \frac{p_k}{\Gamma\!\left(\frac{p_k}{2}\right)} 2^{-\frac{p_k}{2}}
       \int_0^{\infty} (p_k x)^{\frac{p_k}{2}} e^{-\frac{p_k x}{2}} \, dx \\[0.5em]
    &\overset{(\star)}{=} \frac{2}{\Gamma\!\left(\frac{p_k}{2}\right)}
       \int_0^{\infty} y^{\frac{p_k}{2}} e^{-y} \, dy
     = \frac{2\,\Gamma\!\left(\frac{p_k}{2}+1\right)}
            {\Gamma\!\left(\frac{p_k}{2}\right)}
     = p_k.
    \end{align*}
    The equality $(\star)$ follows from the change of variables $y = \frac{p_kx}{2}$. Similarly, using the same limiting arguments, we obtain
    \begin{align*}
    \lim_{N\to\infty} \mathbb{E}[V_{N,k}^2]
    &= \frac{p_k}{\Gamma\!\left(\frac{p_k}{2}\right)} 2^{-\frac{p_k}{2}}
       \int_0^{\infty} (p_k x)^{\frac{p_k}{2}+1} e^{-\frac{p_k x}{2}} \, dx \\
    &= \frac{p_k}{\Gamma\!\left(\frac{p_k}{2}\right)} 2^{-\frac{p_k}{2}}
       \int_0^{\infty} (2y)^{\frac{p_k}{2}+1} e^{-y} \frac{2}{p_k}\, dy \\
    &= \frac{4}{\Gamma\!\left(\frac{p_k}{2}\right)}
       \Gamma\!\left(\frac{p_k}{2}+2\right)
     = p_k(p_k+2).
    \end{align*}
    Therefore,
    \begin{equation*}
        \mathrm{Var}(U_{N,k})
        = \lim_{N\to\infty}
        \left\{\mathbb{E}[V_{N,k}^2] - (\mathbb{E}[V_{N,k}])^2\right\}
        = 2p_k
        \quad \text{as}\quad N\to\infty.    
    \end{equation*}
\end{proof}

\section{Proof of Theorem \ref{thm3}} \label{Appendix_B}
\begin{proof}[Proof of Theorem~\ref{thm3}]
    The proof relies on a central limit theorem for strongly mixing sequences. Since $\{U_{N,k}\}$ is stationary and conditions (C1$^\prime$) and (C2$^\prime$) hold, Corollary~5.1 of \citet{hall2014martingale} implies that it suffices to verify that, for some $\delta>0$,
    \begin{equation}
        \mathbb{E}\!\left[\,\bigl|U_{N,k}-\mathbb{E}[U_{N,k}]\bigr|^{2+\delta}\right] < \infty.
    \end{equation}
    Using the identity in \eqref{eqn:U_Nk}, we obtain
    \begin{align*}
        &\mathbb{E}\!\left[\,\bigl|U_{N,k}-\mathbb{E}[U_{N,k}]\bigr|^{2+\delta}\right]\\
       &= \int_0^\infty 
        \Bigl|\,N\log\!\Bigl(1+\frac{p_k x}{N-p_k-1}\Bigr) - N D_{p_k}(N-2)\Bigr|^{2+\delta}\,
        \frac{p_k(p_k x)^{\frac{p_k}{2}-1}(N-p_k-1)^{\frac{N-p_k-1}{2}}}
            {B(p_k/2, (N-p_k-1)/2)(p_k x+N-p_k-1)^{\frac{N-1}{2}}}
       \, dx \\[0.5em]
        &= B^{-1}\!\left(\frac{p_k}{2}, \frac{N-p_k-1}{2}\right)
        \int_0^\infty 
        \Bigl|\,N\log\!\Bigl(1+\frac{p_k x}{N-p_k-1}\Bigr) - N D_{p_k}(N-2)\Bigr|^{2+\delta} \\
        &\qquad\qquad \times
        \left(\frac{p_k x}{p_k x + N - p_k -1}\right)^{\frac{p_k}{2}-1}
        \left(\frac{N-p_k-1}{p_k x + N - p_k -1}\right)^{\frac{N-p_k-1}{2}-1}
        \frac{p_k (N-p_k-1)}{(p_k x + N - p_k -1)^2}\, dx \\[0.5em]
        &= B^{-1}\!\left(\frac{p_k}{2}, \frac{N-p_k-1}{2}\right)
        \underbrace{\int_0^1 
        \Bigl|\,N\log\!\Bigl(\frac{1}{y}\Bigr) - N D_{p_k}(N-2)\Bigr|^{2+\delta}
        (1-y)^{\frac{p_k}{2}-1} y^{\frac{N-p_k-1}{2}-1}\, dy}_{\eqqcolon \textup{(I)}}. 
    \end{align*}
    The last equality follows by the change of variables $y = \frac{N-p_k-1}{p_kx + N-p_k -1}$.
    By the triangle inequality and the convexity of $x^{2+\delta}$, we obtain
    \begin{align*}
        \textup{(I)}
        &\le 2^{1+\delta}
        \Bigg[
        \underbrace{\int_0^1 |N\log y|^{2+\delta}
        (1-y)^{\frac{p_k}{2}-1} y^{\frac{N-p_k-1}{2}-1}\,dy}_{\eqqcolon\,\textup{(II)}}
        +
        \int_0^1 |N D_{p_k}(N-2)|^{2+\delta}
        (1-y)^{\frac{p_k}{2}-1} y^{\frac{N-p_k-1}{2}-1}\,dy
        \Bigg].
    \end{align*}
    The second integral on the left-hand side of the above inequality is finite since the integrand is a constant multiple of a Beta density, provided that $N \geq \underset{k}{\max} \, p_k + 2$. Hence, it suffices to show that $\textup{(II)}<\infty$.
    We split the integration domain of $\textup{(II)}$ into two parts:
    \begin{align*}
        \textup{(II)}
        =
        \underbrace{
        \int_0^{1/2} |N\log y|^{2+\delta}
        (1-y)^{\frac{p_k}{2}-1} y^{\frac{N-p_k-1}{2}-1}\,dy
        }_{\eqqcolon\,\textup{(III)}}
        + \int_{1/2}^1 |N\log y|^{2+\delta}
        (1-y)^{\frac{p_k}{2}-1} y^{\frac{N-p_k-1}{2}-1}\,dy.
    \end{align*}
    Since $|\log y|\le \log 2$ for $y\in[1/2,1]$, the second integral is clearly finite, provided that $N \geq \underset{k}{\max} \, p_k + 2$. For $\textup{(III)}$, we note that
    \begin{align*}
        \textup{(III)}
        \le
        \sqrt{2} \int_0^{1/2} \{N\log(1/y)\}^{2+\delta}
        \, y^{\frac{N-p_k-1}{2}-1}\,dy
        \le
        \sqrt{2}\int_0^1 \{N\log(1/y)\}^{2+\delta}
        \, y^{\frac{N-p_k-1}{2}-1}\,dy.
    \end{align*}
    Applying the change of variables $t=\log(1/y)$ yields
    \begin{equation*}
        \int_0^1 \{\log(1/y)\}^{2+\delta}
        \, y^{\frac{N-p_k-1}{2}-1}\,dy
        =
        \int_0^\infty t^{2+\delta}
        e^{-t\left(\frac{N-p_k-1}{2}\right)}\,dt
        <\infty,
    \end{equation*}
    provided that $N \geq \underset{k}{\max} \, p_k + 2$. This completes the proof.
\end{proof}

\section{Proof of Theorem \ref{thm4}} \label{Appendix_C}
\begin{proof}[Proof of Theorem~\ref{thm4}]
Note that, under the local alternative \eqref{eqn:local_alter_BILT},
\begin{equation*}
    U_{N,k} \overset{d}{=} 
    N \log \left(1 + \frac{b}{N-b-1} F_k\right),
\end{equation*}
where $F_k$ follows a noncentral $F$ distribution with degrees of freedom $b$ and $N-b-1$, and noncentrality parameter
\begin{equation*}
    \lambda_k = \bdelta_k^\top \bSigma_{kk}^{-1} \bdelta_k.
\end{equation*}
Also note that
\begin{equation*}
    N\left( \frac{b}{N-b-1}F_k - \frac{b^2}{2(N-b-1)^2}F_k^2\right)
    \leq N\log\left(1+ \frac{b}{N-b-1}F_k\right)
    \leq N \frac{b}{N-b-1}F_k.
\end{equation*}
Therefore, for any $N>b+5$, it follows that
\begin{equation*}
    N \left(\frac{b\, \mathbb{E}\!\left[F_k \mid H_1\right]}{N-b-1}
    - \frac{b^2\,\mathbb{E}\!\left[F_k^2 \mid H_1\right]}{2(N-b-1)^2}\right)
    \leq
    N\,\mathbb{E}\!\left[\log\left(1 + \frac{b}{N-b-1}F_k\right)\middle | H_1\right]
    \leq
    N \frac{b\,\mathbb{E}\!\left[F_k \mid H_1\right]}{N-b-1}.
\end{equation*}
Under condition \eqref{eqn:unif_bound_BILT}, we have
\begin{equation*}
    N\,\mathbb{E}\!\left[\log\left(1+\frac{b}{N-b-1}F_k\right)\middle | H_1\right]
    = b\,\mathbb{E}\!\left[F_k \mid H_1\right] + O\!\left(\frac{1}{N}\right).
\end{equation*}
Furthermore, since
\begin{equation*}
    \mathbb{E}\!\left[F_k \mid H_1\right]
    = \frac{(N-b-1)\bigl(b + \bdelta_k^\top \bSigma_{kk}^{-1} \bdelta_k\bigr)}{b\,(N-b-3)},
\end{equation*}
we obtain
\begin{equation*}
    G_k \coloneqq N\,\mathbb{E}\!\left[\log\left(1+\frac{b}{N-b-1}F_k\right)\middle | H_1\right]
    = b + \bdelta_k^\top \bSigma_{kk}^{-1} \bdelta_k + O\!\left(\frac{1}{N}\right).
\end{equation*}
Let $\bar{G} = K^{-1} \sum_{k = 1}^K G_k= b + K^{-1}\bDelta^\top \bDelta + O(1/N)$, and let $\hat{\tau}_{\mathrm{BILT}}^2$ be a consistent estimator of $\tau_{\mathrm{BILT}}^2$. Under the condition (C1$^\prime$) and (C2$^\prime$), we have
\begin{align*}
    \frac{T_{\mathrm{BILT}} - Kb}{\sqrt{K \hat{\tau}_{\mathrm{BILT}}^2}} 
    &= 
    \frac{T_{\mathrm{BILT} }-K\bar{G}}{\sqrt{K \hat{\tau}_{\mathrm{BILT}}^2}} + \frac{K(\bar{G} - b)}{\sqrt{K\hat{\tau}_{\mathrm{BILT}}^2}}\\
    &=
    \frac{T_{\mathrm{BILT}} - K\bar{G}}{\sqrt{K\hat{\tau}_{\mathrm{BILT}}^2}} + \frac{\sqrt{K}(K^{-1} \bDelta^\top \bDelta)}{\sqrt{\hat{\tau}_{\mathrm{BILT}}^2}} + \frac{\sqrt{K}O(1/N)}{\sqrt{\hat{\tau}_{\mathrm{BILT}}^2}}.
\end{align*}
Since we suppose that $K = o(N^2)$, we obtain 
\begin{equation*}
    \frac{T_{\mathrm{BILT}}- Kb}{\hat{\tau}_{\mathrm{BILT}}\sqrt{K}} \overset{d}{\to} N(0,1) +  \frac{\bDelta^\top \bDelta/\sqrt{K}}{\tau_{\mathrm{BILT}}} \quad \text{as} \quad (N, p) \to \infty.
\end{equation*}
This completes the proof.
\end{proof}

\section{Additional Simulation Results} \label{Appendix_D}
\subsection{Asymptotic Normality of the Standardized BILT Statistic}
We investigate the asymptotic normality of the standardized BILT statistic $\tilde{Z}_{\mathrm{BILT}}$ under the null hypothesis. All simulation settings are the same as in Section~\ref{sub_sec_4.1}, except that the underlying covariance structure is set to BD\_0.6. The results are summarized in Figure~\ref{fig:D1}.

We observe that when the block size is $b = 1$, the empirical distribution of $\tilde{Z}_{\mathrm{BILT}}$ deviates noticeably from the standard normal distribution, particularly in the tails. In contrast, the histograms corresponding to BILT with block size $b > 1$ more closely align with the standard normal density.

\begin{figure}[!htb]
    \centering
    \includegraphics[width=\linewidth]{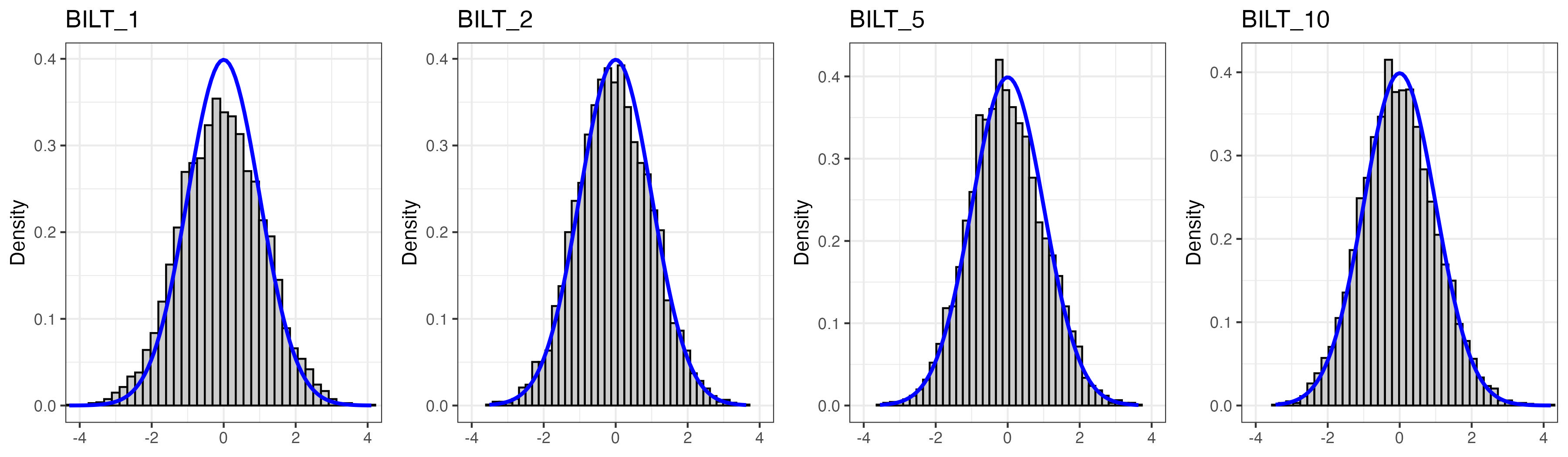}
    \caption{Histogram of the standardized BILT statistics based on 5{,}000 replications, overlaid with the standard normal density, under the BD\_0.6 covariance structure. From left to right, the panels present the results for BILT with block sizes $b = 1, 2, 5,$ and $10$, respectively.}
    \label{fig:D1}
\end{figure}

\subsection{Type~I Error Control}
Next, we examine the asymptotic validity of BILT. The simulation settings are identical to those in Section~\ref{sub_sec_4.2}, except that the sample sizes are set to $n_1 = n_2 = 100$. The results are summarized in Figure~\ref{fig:D2}.

Figure~\ref{fig:D2} exhibits patterns similar to those observed in Figure~\ref{fig:2}. In particular, when the block size is $b = 1$, the type~I error is controlled at the nominal level under relatively weak dependence structures. However, under relatively strong dependence, it fails to maintain type~I error control across all dimensions $p$. In contrast, for $b > 1$, the type~I error tends to approach the nominal level as $p$ increases. Consistent with the previous findings, BILT with $b = 2$ provides stable and reliable type~I error control across all simulation settings.

\begin{figure}[!htb]
    \centering
    \includegraphics[width=\linewidth]{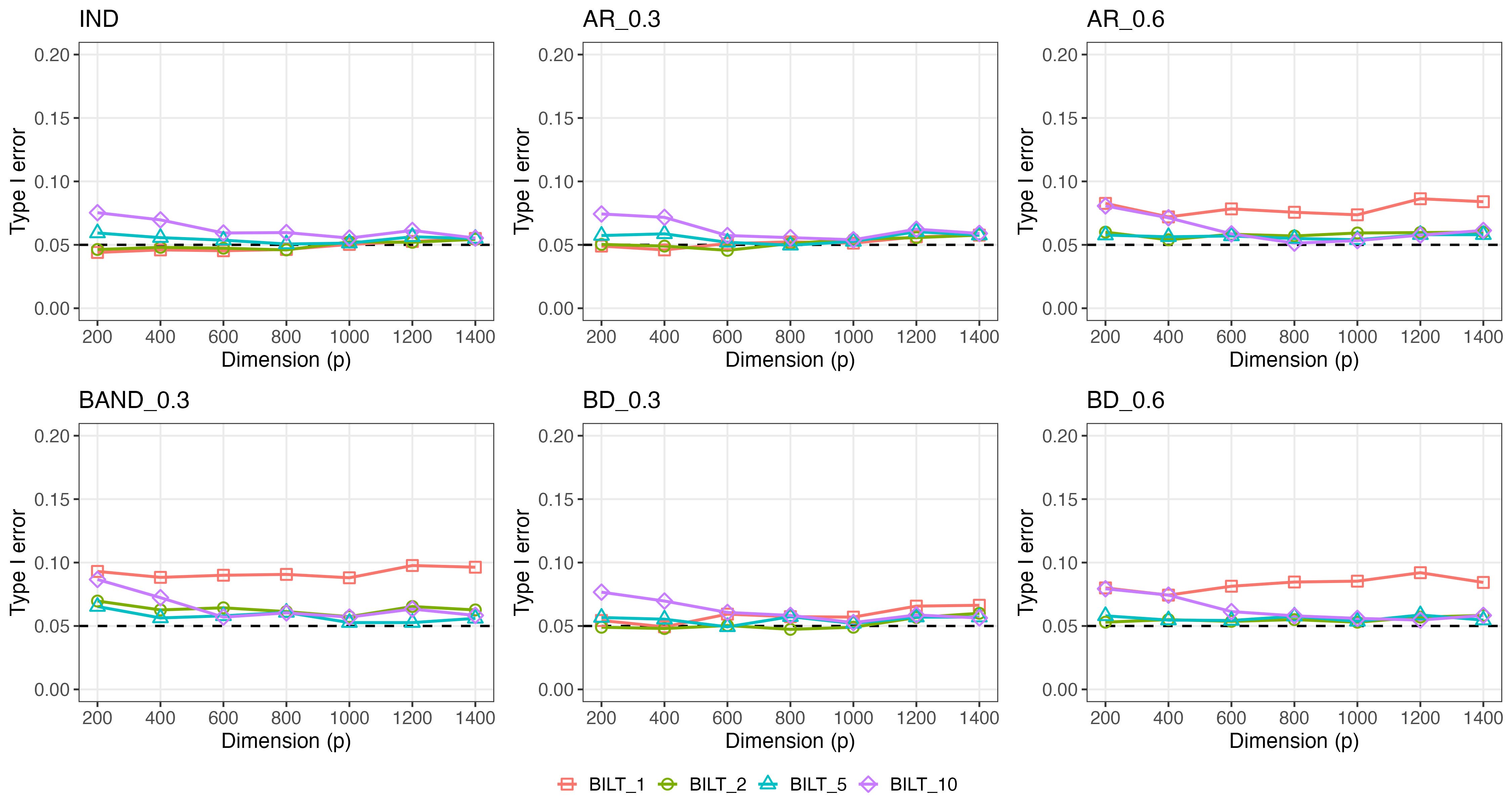}
    \caption{Type~I error of BILT with $n_1 = n_2 = 100$. Each panel corresponds to a covariance structure. Within each panel, results for block sizes $b \in \{1, 2, 5, 10\}$ are displayed. The x-axis represents the dimension $p$, and the y-axis represents the Type~I error obtained from 3{,}000 replications. The horizontal dashed line indicates the nominal Type~I error level of $0.05$.}
    \label{fig:D2}
\end{figure}

\subsection{Power Curve Against Signal Magnitude}
Now, we examine the asymptotic power of BILT as the signal magnitude varies. The simulation settings are identical to those in Section~\ref{sub_sec_4.3}, except that the sample sizes are increased to $n_1 = n_2 = 100$. The results are summarized in Figure~\ref{fig:D3}.

Figure~\ref{fig:D3} exhibits patterns consistent with the earlier findings. Power increases monotonically with the signal magnitude. Under IND, all block sizes yield nearly identical power, indicating no loss from enlarging $b$. In contrast, under dependent covariance structures, increasing the block size improves power, particularly when dependence is strong (BAND\_0.3, AR\_0.6, BD\_0.6). The improvement remains evident even for $b = 2$, further supporting its use as a practical default. Compared with Figure~\ref{fig:3}, where $n_1 = n_2 = 50$, increasing the sample size leads to uniformly higher power for a given signal magnitude.

\begin{figure}[!htb]
    \centering
    \includegraphics[width=\linewidth]{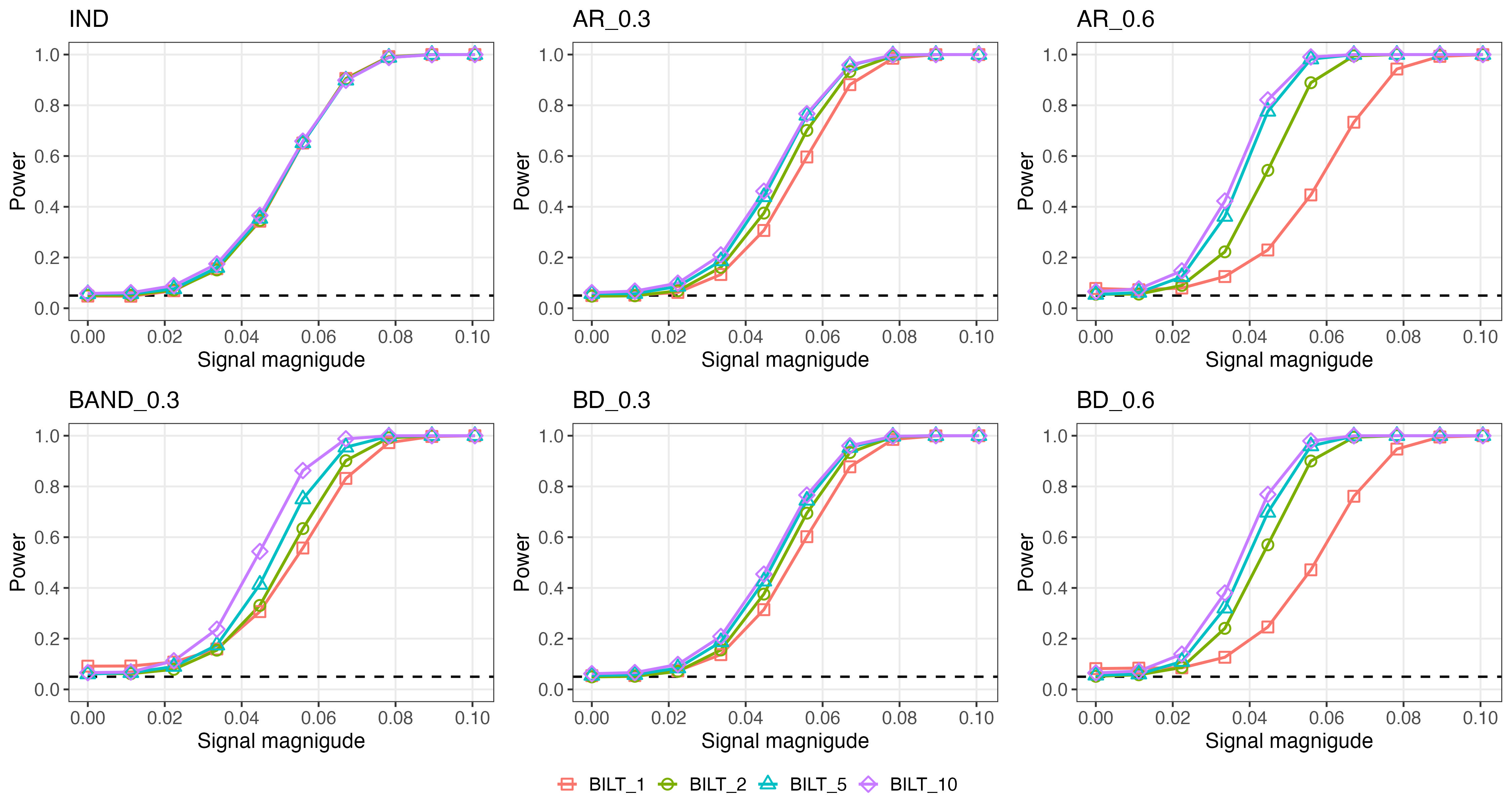}
    \caption{Power of BILT with $n_1 = n_2 = 100$ and $p = 500$. Each panel corresponds to a covariance structure. Within each panel, results for block sizes $b \in \{1, 2, 5, 10\}$ are displayed. The x-axis represents the signal magnitude $\delta/\sqrt{p}$, and the y-axis represents the power obtained from 3{,}000 replications. The horizontal dashed line indicates the nominal Type~I error level of $0.05$.}
    \label{fig:D3}
\end{figure}

\subsection{Power Curve Against Non-Null Proportion}
Finally, we examine the asymptotic power of BILT under varying non-null proportions. The simulation settings are identical to those in Section~\ref{sub_sec_4.4}, except that the sample sizes are increased to $n_1 = n_2 = 100$. The results are presented in Figure~\ref{fig:D4}.

We can observe the patterns consistent with the Figure~\ref{fig:4}. Overall, power increases as the non-null proportion increases. Under IND, all block sizes yield nearly identical power. In contrast, under dependent covariance structures, power improves with larger block sizes, particularly when dependence is strong. The improvement remains evident even for $b = 2$. Compared with Figure~\ref{fig:4}, where $n_1 = n_2 = 50$, increasing the sample size leads to uniformly higher power for a given non-null proportion.

\begin{figure}[!htb]
    \centering
    \includegraphics[width=\linewidth]{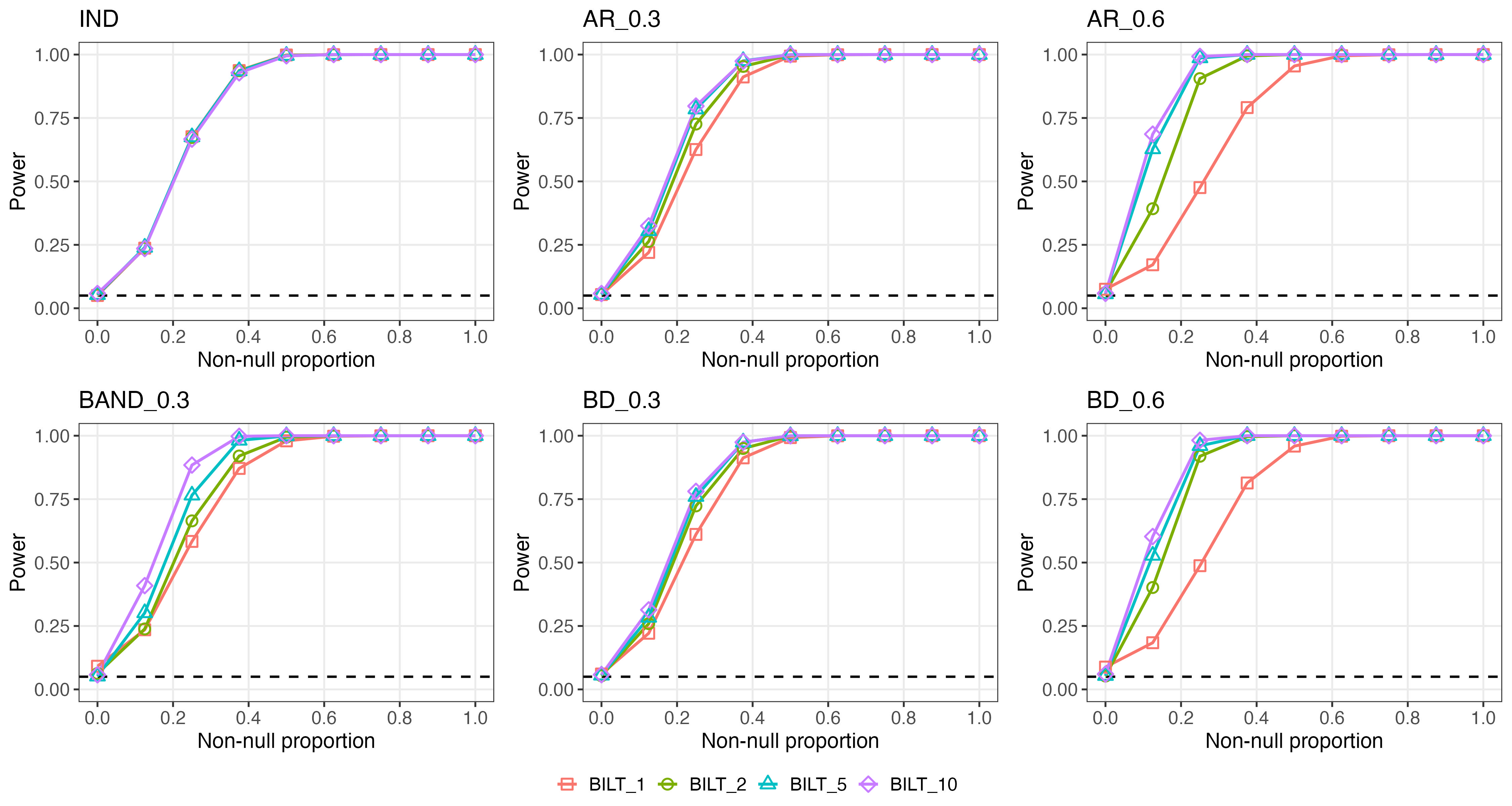}
    \caption{Power of BILT with $n_1 = n_2 = 100$ and $p = 1{,}000$. Each panel corresponds to a common covariance structure. Within each panel, results for block sizes $b \in \{1, 2, 5, 10\}$ are displayed. The x-axis represents the non-null proportion, and the y-axis represents the power obtained from 3{,}000 replications. The horizontal dashed line indicates the nominal Type~I error level of $0.05$.}
    \label{fig:D4}
\end{figure}

\end{document}